 \documentclass[final,5p,times,twocolumn]{elsarticle}

\usepackage{amssymb}
\usepackage{graphicx}
\usepackage[caption=false,font=footnotesize]{subfig}

\usepackage{amsmath}
\usepackage{amsthm}
\usepackage{xcolor}
\usepackage{array}
\usepackage{multirow,tabularx}
\usepackage{mathtools}

\DeclarePairedDelimiter\ceil{\lceil}{\rceil}
\newtheorem{theorem}{Theorem}
\newtheorem{problem}{Problem}
\newcommand{\argmax}{\operatornamewithlimits{argmax}}
\usepackage{algorithm,algorithmic}
\usepackage[hyphens]{url}

\journal{Ad Hoc Networks}

\begin{document}

\begin{frontmatter}

\title{Maximum Lifetime Convergecast Tree in Wireless Sensor Networks}

\author[label1]{Jobish John\corref{cor1}}
\ead{jobish.john@ee.iitb.ac.in}

\author[label1]{Gaurav S. Kasbekar}
\ead{gskasbekar@ee.iitb.ac.in}

\author[label1]{Maryam Shojaei Baghini}
\ead{mshojaei@ee.iitb.ac.in}

\address[label1]{Department of Electrical Enginerring, Indian Institute of Technology Bombay, India}

\cortext[cor1]{ Corresponding author}

\begin{abstract}
We study the problem of building a maximum lifetime data collection tree for periodic convergecast applications in wireless 
sensor networks. We experimentally observe that if two nodes transmit the same number of data packets, the amount of energy consumption of the nodes
is approximately the same even if the payload lengths of the transmitted packets are different. This is because the major energy consumption during a packet transmission arises from radio start-up and medium access control overhead.
Our formulated lifetime maximization problem captures the energy expenditure due to message transmissions/ receptions in terms of the number 
of data \emph{packets} transmitted/ received, in contrast to prior works, which consider the number of data \emph{units} (amount of sensor data generated by a node) transmitted/ received. 
Variable transmission power levels of the radio and accounting for the sensor energy consumption are other factors that make
our problem formulation different from those in prior work.
We prove that this problem is NP-complete by reducing the set cover problem to it and propose an algorithm to solve it. The performance of the proposed algorithm is experimentally evaluated using Jain's fairness index as a metric by implementing it on an actual  testbed consisting of 20 sensor nodes and compared with those of the widely used shortest path tree and random data collection tree algorithms. The energy consumption of different nodes under the proposed algorithm are shown to be more balanced than under the shortest path tree and random data collection tree algorithms. Also, the performance of the proposed algorithm in large networks is studied through simulations and is compared with those of the state-of-the-art RaSMaLai algorithm, the shortest path tree, minimum spanning tree, and random tree based data collection schemes. Our simulations show that the proposed algorithm provides a significantly higher network lifetime compared to all the other considered data collection approaches.
\end{abstract}


\end{frontmatter}

\section{Introduction}
\label{}

Wireless  sensor networks (WSN) are being extensively deployed for numerous monitoring applications such as environmental 
monitoring~\cite{lombardo2017wireless,yang2010design,barrenetxea2008sensorscope}, structural monitoring~\cite{pakzad2008design} and
agricultural monitoring~\cite{sharma2019maximization,langendoen2006murphy,ojha2015wireless,jobishjohnEAI}.
In a WSN, a large number of small, low-cost, resource-constrained devices called ``sensor nodes'' collectively sense or monitor an area of interest.
In most WSN monitoring applications, each sensor node reports its collected data to a decision center (also known as base station) in a multi-hop fashion, often via wireless transmissions along the edges of a tree.
This process of collecting data from all the sensors at the base station is known as \emph{convergecast}~\cite{fastcollection}. 
Sensor nodes are powered using small batteries 
and recharging of these batteries is difficult since often, sensor nodes are deployed in harsh, 
inaccessible areas and they are expected to operate for a long time with minimal human intervention. So minimization of the energy
consumed by sensor nodes is an important objective.

Network lifetime of a WSN can be defined in different ways~\cite{mak2009long}.
The time until the first node dies due to battery depletion and the time until a certain fraction of nodes get depleted are 
two commonly used notions for the network lifetime~\cite{mak2009long}.
In our work, by ``lifetime'' of the network, we mean the time until the first node in the network fails due to battery depletion.
%In our work, we consider the time until the first node dies as the network lifetime.
The wireless communication strategies in a WSN need to be selected in such a way that the lifetime of the WSN is maximized.
Different approaches are used for energy conservation of the nodes. Efficient duty cycling~\cite{anastasi2009extending}, 
data aggregation/ compression schemes~\cite{xu2010delay,min2012data}, load balancing among the nodes~\cite{toumpis2009load,junbin2010efficient,imon2015energy}, 
clustering techniques~\cite{yarinezhad2020increasing,younis2006node} and energy efficient routing protocols~\cite{pantazis2012energy} are some of them.
Also, the data collection or convergecast operation in sensor networks needs to be carried out in an energy efficient manner in order to maximize
the lifetime of the network.

Various approaches have been proposed in the literature for performing the convergecast operation so as
to maximize the lifetime of the sensor network. These studies can be broadly classified into tree-based approaches, 
integer programming approaches, time allocation approaches and flow based schemes~\cite{lin2017approximation}. 
In tree-based data collection approaches~\cite{buragohain2005power, wu2010constructing,huang2012load}, a data collection tree constructed in an energy efficient manner
is used for the convergecast. The integer programming approach formulates the lifetime maximization problem as an integer
0-1 programming problem~\cite{lee2012improved}. A non-integer solution obtained using the linear programming relaxation approach is 
converted to an integer solution.
Time allocation schemes maximize the network lifetime by using different pre-calculated trees for 
data collection in a time-multiplexed manner~\cite{berman2004power,xiong2009maximize,wen2007tree}. 
Flow based approaches consider the data sent by the sensor nodes as a network flow and the lifetime maximization problem is 
formulated as a network flow problem~\cite{ordonez2004optimal,luo2011maximizing}.

The problem of building the data collection tree, in the tree-based data collection approach, so as to 
maximize the network lifetime is a well studied problem in the literature~\cite{lin2017approximation,buragohain2005power,tan2003power,wan2016construction,zhou2016maximizing}. Maximum lifetime tree construction has been shown to be
an NP-complete problem under the aggregated/ non aggregated convergecast models considered in~\cite{buragohain2005power}. 
There are several works which propose both approximation algorithms as well as heuristic approaches for the construction of a
maximum lifetime data collection tree. Various cross-layer optimization techniques are reviewed in~\cite{fei2016survey} which try to jointly optimize the 
routing, power allocation, and node scheduling schemes. A maximum lifetime tree construction approach for the fully aggregated data collection model is 
proposed and the NP-completeness of the problem is proved in~\cite{wu2008construction}. 
In the proposed algorithm, the authors start from an arbitrary tree and iteratively try to reduce the load on the bottleneck nodes 
(the nodes which are expected to die fast). Load balancing/ load switching is one of the dominant approaches used for building the maximum lifetime data collection tree~\cite{imon2013rasmalai,imon2015energy}.
An algorithm for maximum lifetime tree construction under the data collection without aggregation model is considered in~\cite{junbin2010efficient}.
The authors formulate the maximum lifetime tree construction problem as a min-max spanning tree construction problem. The algorithm starts with an arbitrary tree and iteratively transfers descendant nodes from nodes having high weights to those having low weights.
The weight of a node in a data collection tree is a function of its remaining battery energy and the number of its descendant 
nodes in the considered tree. In~\cite{imon2015energy}, the authors describe a randomized switching algorithm for load balancing among the nodes
to construct a maximum lifetime data collection tree.
Most of the maximum lifetime tree construction approaches in the literature assume that all the nodes in the network transmit at a 
fixed transmission power~\cite{wu2008construction,imon2013rasmalai}. In~\cite{lin2017approximation}, the authors consider variable transmission power levels and 
propose an approximation algorithm for the maximum lifetime data collection tree construction problem. 

Our work differs from prior work in the following respects. We formulate the problem of building a maximum lifetime data collection tree for periodic convergecast applications by considering the energy needed for the transmission/ reception of \emph{data packets} in the problem formulation; in contrast, in formulations in prior literature, the energy
 needed for the transmission/ reception of \emph{data units} is considered. 
Several data units can be combined to form a data packet. The experimental observations described in Section~\ref{uniqueaspects_motivations}
 reveal that the energy drainage of a node is mainly dependent on the \emph{number of  transmissions/ receptions of data packets} and not the  \emph{number of  transmissions/ receptions of data units}. Our work also takes into account the \emph{energy expenditure of sensors} in addition to \emph{different transmission power levels} for the radio. In contrast, most of the prior works ignore the sensor energy consumption and assume that the radio transmission power is fixed. 
In our paper, the performance of the proposed algorithm for building the maximum lifetime data collection tree is evaluated through an \emph{actual testbed implementation}, whereas a majority of the proposed lifetime maximization approaches in the literature are evaluated only via simulations.

The major contributions of this paper are as follows:
\begin{itemize}
   \item {We formulate the maximum lifetime data collection tree construction problem by considering the energy needed
 for the transmission/ reception of  \emph{data packets} instead of \emph{data units}. }
 \item {Variable transmission power levels of the wireless radio and the energy expenditure for generating sensor data are taken into account in our problem formulation.}
 \item {We prove that the above problem is NP-complete and propose an algorithm for solving it.}
  \item {The proposed algorithm is implemented in an actual WSN testbed having $20$ sensor nodes (TelosB~\cite{telosb_mote} motes programmed using TinyOS 2.1.2~\cite{tinyos}) 
 and its performance is compared with those of the widely used shortest path tree (SPT) based and random tree (RDCT) based data collection approaches.}
 \item {Our experimental evaluation demonstrates that a more balanced discharge of batteries among the nodes occurs 
 when the proposed algorithm  is used for building the maximum lifetime data collection tree as compared to the discharge under the SPT based and RDCT based approaches. Hence, our proposed algorithm results in an improved network lifetime.}
\item {The performance of the proposed algorithm in large networks is evaluated through simulation studies. Our simulations show that the proposed algorithm provides significantly higher network
 lifetime when compared with the state-of-the-art Randomized Switching for Maximizing Lifetime (RaSMaLai) algorithm~\cite{imon2015energy} as well as the SPT, minimum spanning tree (MST), and RDCT based data collection schemes.}
\end{itemize}
 
 The rest of the paper is organized as follows. In Section~\ref{network_model_pblm_statement}, the considered network model and 
 the problem formulation are described. Section~\ref{uniqueaspects_motivations} discusses the unique aspects of the problem and the 
 motivating factors. The hardness of the problem is characterized  and the proposed algorithm is described in Section~\ref{hardness_solution}. 
 Section~\ref{Performance_evaluation} presents a  performance evaluation of the proposed algorithm and the paper is concluded in Section~\ref{conclusion}.
 
\section{Network model and problem formulation}
\label{network_model_pblm_statement}

\subsection{Network Model}
\label{Network_architecture}

In sensor networks which are deployed for any periodic data collection application such as agricultural farm  
monitoring, \emph{convergecast} is the most common operation~\cite{fastcollection}.
That is, data from all the individual sensor nodes are collected at a sink node via transmissions along the edges of a tree. The sink node is connected to a powerful
data logging device called base station (e.g., a laptop) and is assumed to have strong data processing capabilities and infinite energy resources. 

\begin{figure}[h]
\centering
\subfloat[Packet relay model]{\includegraphics[scale=0.35]{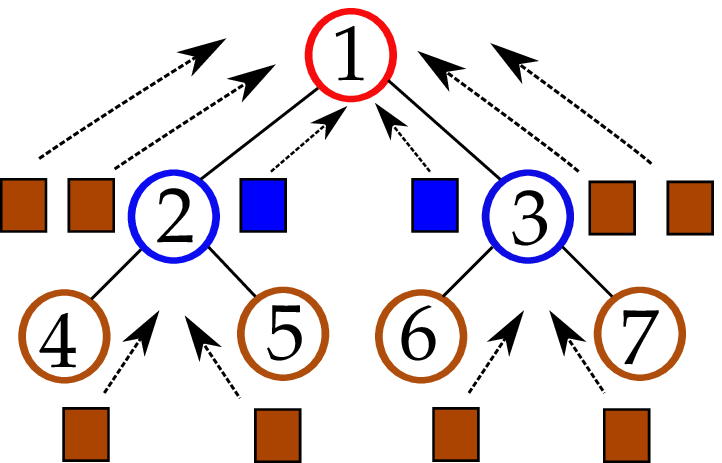}%
\label{dct_seperate_pkts}}
\hfil
\subfloat[Complete aggregation model]{\includegraphics[scale=0.35]{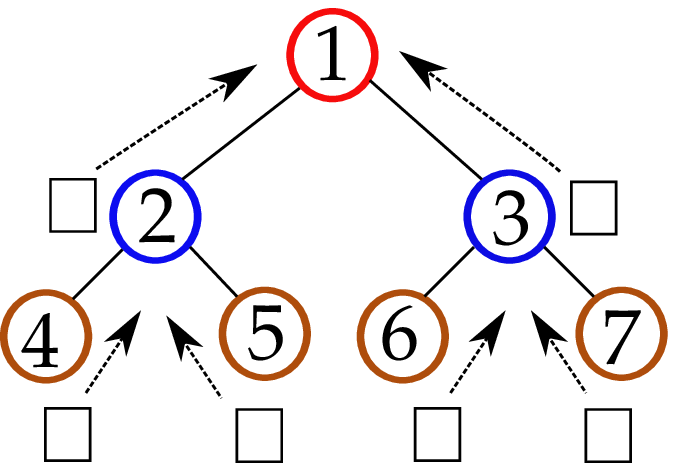}%
\label{dct_justonepkt}}
\hfil
\subfloat[Piggyback aggregation model]{\includegraphics[scale=0.35]{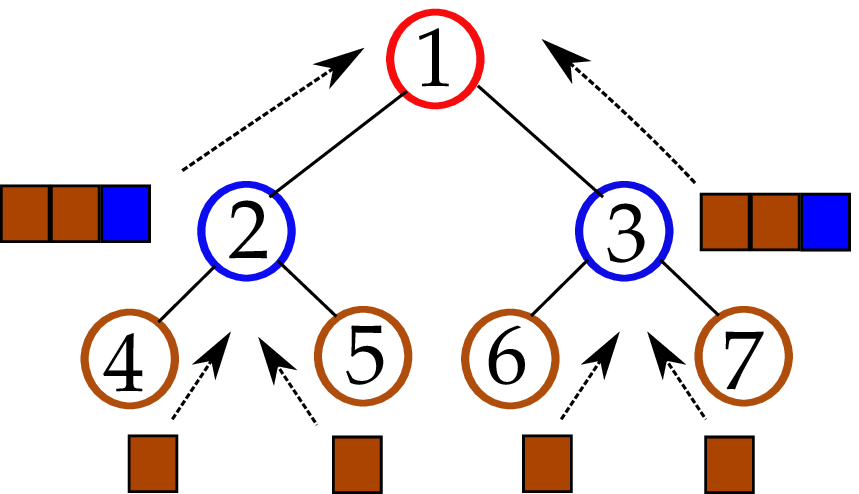}%
\label{dct_piggyback}}
\caption{Data collection along a convergecast tree}
\label{datacollection_schemes_wsn}
\end{figure}

There are mainly three approaches which are widely used for data collection in WSN applications, as shown in Fig.~\ref{datacollection_schemes_wsn}.
In the first approach (Fig.~\ref{dct_seperate_pkts}), each sensor node sends its data unit (the generated sensor data) as a data packet to its parent 
and each parent node relays the received packets from each of its children as separate packets towards the sink node along 
the data collection tree. Also, the sensor data generated by the parent node itself is forwarded as a separate data packet along the tree.
We call this network model as ``packet relay model'' (PRM).
The other two approaches employ aggregation methods to reduce the number of message transmissions occurring in the network. 
In the second approach (Fig.~\ref{dct_justonepkt}), each node receives data (e.g., 1 byte) from all of its children and applies an aggregation technique
like taking average, minimum, maximum etc. on the collected dataset (received data as well as its own generated data) and forwards 
the aggregated result (also 1 byte in the considered example) to its parent node. We call this network model as ``complete aggregation model'' (CAM).
This method is used mainly in dense deployments to sense 
parameters which have a high spatial correlation.
In the third approach, each node collects the sensor data from its children and concatenates them along with
its own sensed data to form a single data packet and forwards it to its parent node as shown 
in Fig.~\ref{dct_piggyback}. We call this network model as ``piggyback aggregation model'' (PAM). We consider a sensor network of nodes that are sparsely  deployed in an area-- such a sensor network is often deployed for covering a large area with a small number of nodes for  cost-effectiveness. Hence we consider that all the sensor readings are equally important and each reading needs to be transmitted in its entirety to the sink node. So 
the data collection in our network model is of the type PAM as shown in Fig.~\ref{dct_piggyback}. Note that the PAM has a lower packet header overhead than the PRM.

\begin{figure}[h]
 \centering
\includegraphics[scale=0.55]{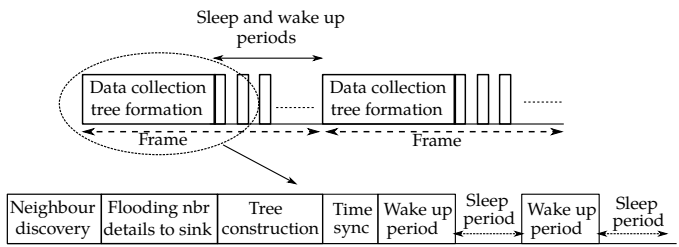}
 \caption{Sequence of operations occurring in the network}
 \label{periodic_operations_toplevel}
\end{figure}

For data collection in sensor networks, a backbone structure (tree) to route the sensed data  from each node to the sink node needs to be constructed.
The various operations occurring in the network are shown in Fig.~\ref{periodic_operations_toplevel} (for details see \cite{jobishjohnEAI}).
Time is divided into frames and each frame consists of the ``data collection tree formation phase'' followed by periodic data collection through the constructed tree as shown in Fig.~\ref{periodic_operations_toplevel}.
The ``data collection tree formation'' phase constructs the data collection tree through which each node reports its sensed data to the sink node.
Once the data collection tree is built, each node follows a synchronized periodic sleep wake-up scheme for data collection.
This helps to improve the lifetime of the network since during each ``sleep period'', the wireless transceiver and sensors in all nodes
remain in the ``OFF'' state to save energy. During each ``wake-up period'' of a frame, data collection is performed using the tree constructed in the 
``data collection tree formation'' phase of the frame.  
Each ``data collection tree formation'' phase includes various stages like neighbour discovery, flooding the neighbour details to the sink, running the tree 
construction algorithm at the sink node and synchronizing all the nodes in the network to the sink node. 
In the ``neighbour discovery'' phase, each node in the network finds its neighbours, 
along with the required transmission power levels to reliably communicate with them. Each node assigns  ``edge costs'' to 
all the communication links connecting it to its neighbours which are functions of the minimum transmission power levels required to deliver packets 
reliably to them. After neighbor discovery, the list of neighbours and edge weights of each node are
sent to the sink node through ``flooding''. The sink node builds the maximum lifetime data collection tree 
from the collected information and
synchronizes the clocks of all the nodes along the edges of the tree. Then, each node periodically reports its data to the sink node.

As stated earlier, to increase the lifetime of the network, nodes follow a periodic sleep and wake-up schedule (see Fig.~\ref{periodic_operations_toplevel}). 
Nodes are in energy saving sleep mode for a large fraction of the time.
During each wake-up period, each node senses parameters using its sensors, receives data from its child nodes, sends the data to its parent node and goes to the sleep stage. 

\subsection{Problem Formulation}
\label{pblm_statement}

Consider a sensor network consisting of $n$ sensor nodes $\{v_1,v_2, ..., v_n\}$ and a sink node $v_0$ which are
deployed for any data monitoring application.
After the ``data collection tree formation phase'' (see Section~\ref{Network_architecture}), each sensor node reports its sensed data
(one data unit, which consists of $l$ bytes) to the sink node in each wake-up period along the data collection tree.
Let $G=(V,E)$ be the undirected connectivity graph representing the sensor network, 
where $V = \{v_0,v_1,v_2, ..., v_n\}$ and $E$ represents all the communication links present in the network. Between every pair of neighbouring
nodes $u$ and $v$, there exists an edge $e_{u-v} \in E$ with cost $c_{e_{u-v}}$, which represents the 
transmission energy required for sending one data packet along the edge $e_{u-v}$ (see Section~\ref{Network_architecture}).
Each sensor node $u$ is equipped with a Lithium-Ion battery which has
a remaining energy $B_u$ at the beginning of the frame under consideration. The sink node $v_0$ is assumed to have unlimited energy, $i.e.,\ B_0 = \infty$.

%By ``lifetime'' of the network, we mean the time until the first node in the network fails due to battery depletion.
Fig.~\ref{dc_example} shows an instance of a connectivity graph $G$ and two of its possible data collection trees $T_1$ and $T_2$.
Our objective is to find a data collection tree in each frame such that the lifetime of the network is maximized. Henceforth, we focus on
a single frame and study the problem of finding the best tree, say $T^*$, from the point of view of maximizing the network lifetime. 
Table~\ref{list_notations} lists the various notations used in this paper.

\begin{figure}
\centering
\subfloat[Graph $G$]{\includegraphics[scale=0.35]{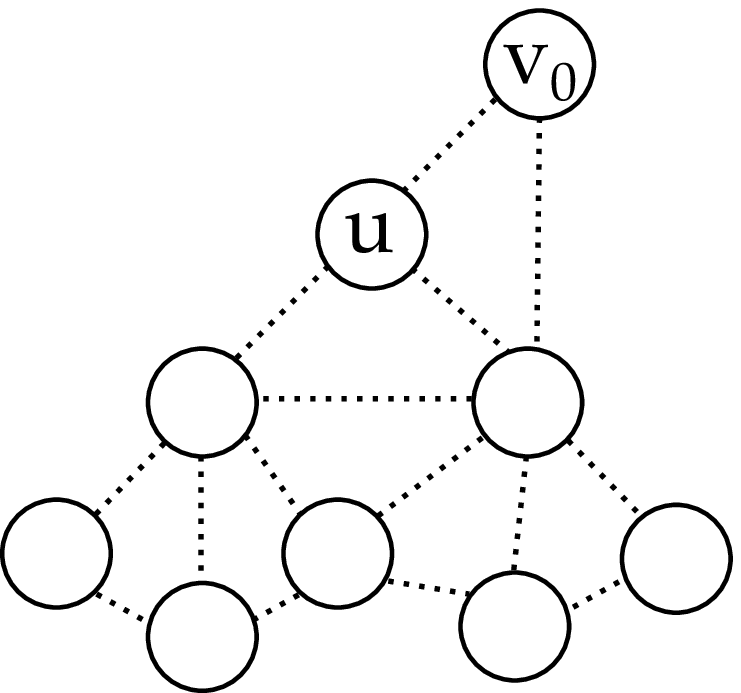}%
\label{dct_graph}}
\hfil
\subfloat[Tree $T_1$]{\includegraphics[scale=0.35]{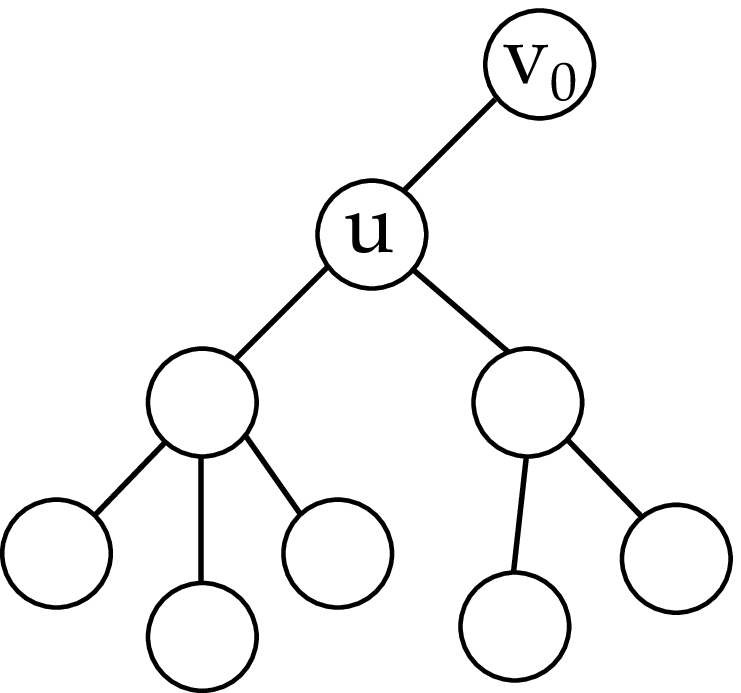}%
\label{dct_example1}}
\hfil
\subfloat[Tree $T_2$]{\includegraphics[scale=0.35]{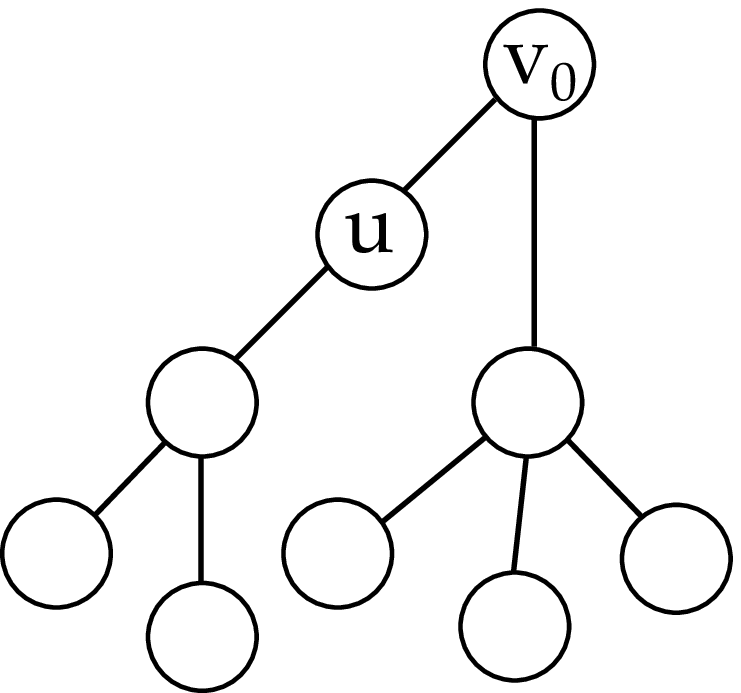}%
\label{dct_example2}}
\caption{An example of a connectivity graph and two possible data collection trees}
\label{dc_example}
\end{figure}
%----------------------------------------------------------------------------------------------------

Recall that after a data collection tree, say $T$, is constructed, each node sends its sensed data to the sink node $v_0$ via the edges of $T$ in each wake-up period. In particular, during each wake-up period, each node $u \neq v_0$ receives all the data coming from its children in the tree $T$ and sends the concatenated data packet which contains the received data along with its own generated sensor data to its parent node,  say $p_u(T)$, in the data collection tree $T$. For any node $u \neq v_0$ in the  tree $T$, let $n_u^r(T)$ be the number of data packets
node $u$ receives from its children in the tree $T$ in a wake-up period. Let
$n_u^t(T)$ be the number of data packets node $u$ transmits to its parent  $p_u(T)$ in a wake-up period.
 Let $e_r$ be the energy spent by node $u$ for receiving 
a data packet from one of its children and $k_u$ be the energy required by  node $u$ to generate data from its own sensors;
recall that node $u$ requires energy $c_{e_{u-p_u(T)}}$ to transmit one data packet to $p_u(T)$.
Hence, the total energy consumed by node $u$ in a wake-up period is $n_u^r (T)\ . \ e_r + n_u^t (T)\ . \ c_{e_{u-p_u{(T)}}} + k_u$.
The number of wake-up periods for which node $u$ can work before its battery gets depleted, if tree $T$ is used in each wake-up period, is given by: 
\begin{equation}
 L_u^T = \dfrac{B_u}{n_u^r(T) \ . \ e_r + n_u^t(T) \ . \ c_{e_{u-p_u{(T)}}} + k_u}
\end{equation}

\begin{table}[h]
\caption{Notations used in this paper}
\label{list_notations}
\centering
\begin{tabular}{|>{\centering\arraybackslash}p{1cm}|>{\centering\arraybackslash}p{7cm}|}
\hline
Notation&Meaning\\
\hline

$n$&The number of sensor nodes in the network\\
\hline
$G = (V,E)$& Undirected connectivity graph representing the sensor network with $V = \{v_0,v_1,v_2, ..., v_n\}$ being the set of nodes in the network and $E$ being the set of edges connecting them\\
\hline
$e_{u-v}$& Edge connecting the nodes $u$ and $v$\\
\hline
$c_{e_{u-v}}$& Cost of the edge $e_{u-v}$ and it represents the amount of energy required to transmit one data packet along the edge $e_{u-v}$ \\
\hline
$B_u$ & Remaining energy of the node $u$ \\
\hline
$n_u^r(T)$ & Number of data packets received by node $u$ from its children in a wake-up period when tree $T$ used \\
\hline
$n_u^t(T)$ & Number of data packets transmitted by node $u$ to its parent in a wake-up period when tree $T$ used \\
\hline
$p_u(T)$ & Parent of node $u$ in the tree $T$ \\
\hline
$k_u$ & Energy consumed by node $u$ for generating data from its sensors \\
\hline
$L_u^T$ & Lifetime of node $u$ when data collection tree $T$ used \\
\hline
$L_{max}$ & Maximum lifetime of the network \\
\hline
$T^*$ & The optimal data collection tree \\
\hline
$\Omega$ & The set all possible data collection spanning trees of the connectivity graph $G = (V,E)$ \\
\hline
$e_r$ & Amount of energy required for receiving one data packet \\
\hline
$C_{avg}$ & Average packet transmission energy (considering different possible transmission power levels) \\
%$C_{avg}$ & Average of different possible data packet transmission energies \\
\hline
$\alpha_u$ & Current load of node $u$ \\
% \hline
% $e_r$ & Average energy required to receive one data packet \\
\hline
$l$ & Number of bytes in one data unit \\
\hline
$\beta$ & Maximum number of bytes that can be included in one data packet \\
\hline
$h_u$ & Hop distance of node $u$ from the sink node  \\
\hline
\end{tabular}
\end{table}

We define the lifetime of the network, say $L(T)$, under tree $T$ as the number of wake-up periods until a node $u\in V$ depletes
its energy if tree $T$ is used for data collection. So:
\begin{eqnarray}
L(T)  & = & \min_{u \in V} \ \{ L_u^T\} \nonumber \\ 
& = & \min_{u\in V} \left \{ \dfrac{B_u}{n_u^r(T) \ . \ e_r + n_u^t(T) \ . \ c_{e_{u-p_u{(T)}}} + k_u} \right \}. \label{optequation2}
\end{eqnarray}
The optimal data collection tree, say $T^*$, is one that has the maximum lifetime among the set of all possible spanning trees, say $\Omega$,  
and its lifetime is given by:
\begin{equation}
\label{asd}
L_{max} = L(T^*)=  \max_{T \in \Omega} L(T).
\end{equation}
By (\ref{optequation2}) and (\ref{asd}):
\begin{equation}
\label{optimaleqn_final}
L_{max}=L(T^*)=\max_{T \in \Omega} \ \min_{u\in V} \left \{ \dfrac{B_u}{n_u^r(T) \ . \ e_r + n_u^t(T) \ . \ c_{e_{u-p_u{(T)}}} + k_u} \right \}.
\end{equation}

Our objective is to design an algorithm that finds the tree, $T^*$, with the maximum lifetime.

\begin{problem}
\label{pblm1}
Find the optimal data collection tree, $T^*$, which provides the maximum lifetime $L_{max}$ in \eqref{optimaleqn_final}. 
\end{problem}

\section{Motivating factors and unique aspects of the lifetime maximization problem }
\label{uniqueaspects_motivations}
In this section, we explain the motivations behind studying the specific problem formulated in Section~\ref{pblm_statement} and the unique 
aspects of the considered problem. 

\subsection{Effect of payload length and number of packet transmissions on the energy consumption}
\label{effect_payloadlength}
To better understand the variation of the energy expenditure of the nodes with variations in the data payload size as well as the number of 
packet transmissions, we conducted a set of experiments in our testbed. Our sensor nodes use TelosB motes~\cite{telosb_mote} programmed using TinyOS 2.1.2~\cite{tinyos}. TelosB uses CC2420 as the wireless transceiver and it supports different transmission power levels~\cite{crossbow2010telosb}.
Fig.~\ref{current_consumption_wrtpayload} shows the current consumption during the transmission of a single data packet for different payload sizes and transmission
power levels. The maximum limit for the data payload length in a single packet is 114 bytes.
From Fig.~\ref{label_maxpayload}, it is clear that the change in current consumption for a payload length of 114 bytes due to the adjustment in 
transmission power level from its minimum to maximum level is approximately 10 mA for a period of 5 ms. This time reduces to approximately 2 ms (Fig.~\ref{label_minload}) if the data payload reduces to 10 bytes. 
The difference in energy consumption due to a change in the data payload size from 114 bytes to 10 bytes when data is transmitted with 
the maximum power is approximately 100 $\mu$J. This is \emph{negligible when compared with the total energy consumption of a data packet transmission}.
For example, a data packet with payload size 114 bytes consumes 1.2 $m$J when transmitted with maximum transmission power.

\begin{figure}[h]
\centering
\subfloat[Payload length - 114 bytes]{\includegraphics[scale=0.24]{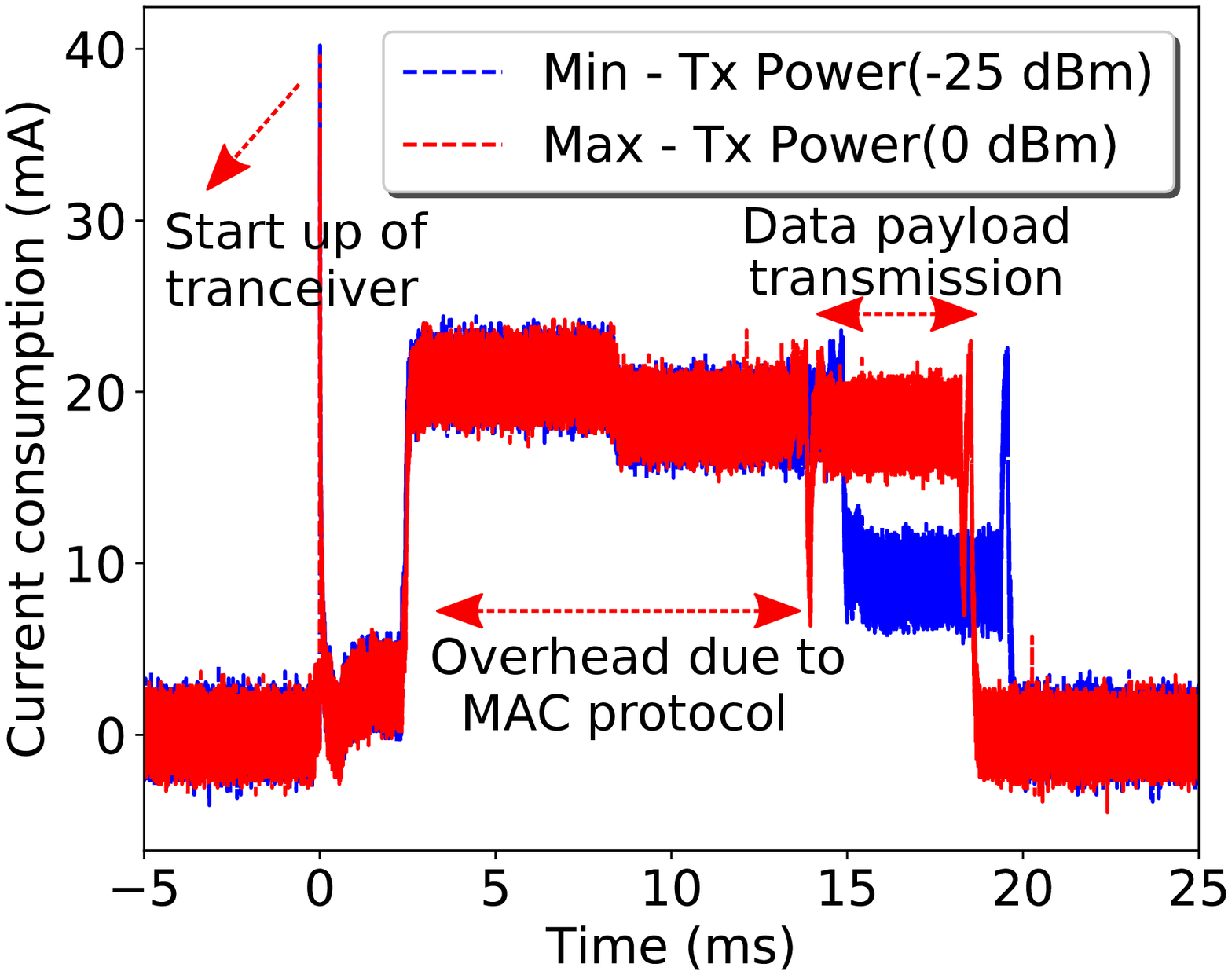}%
\label{label_maxpayload}}
\hfil
\subfloat[Payload length - 10 bytes]{\includegraphics[scale=0.225]{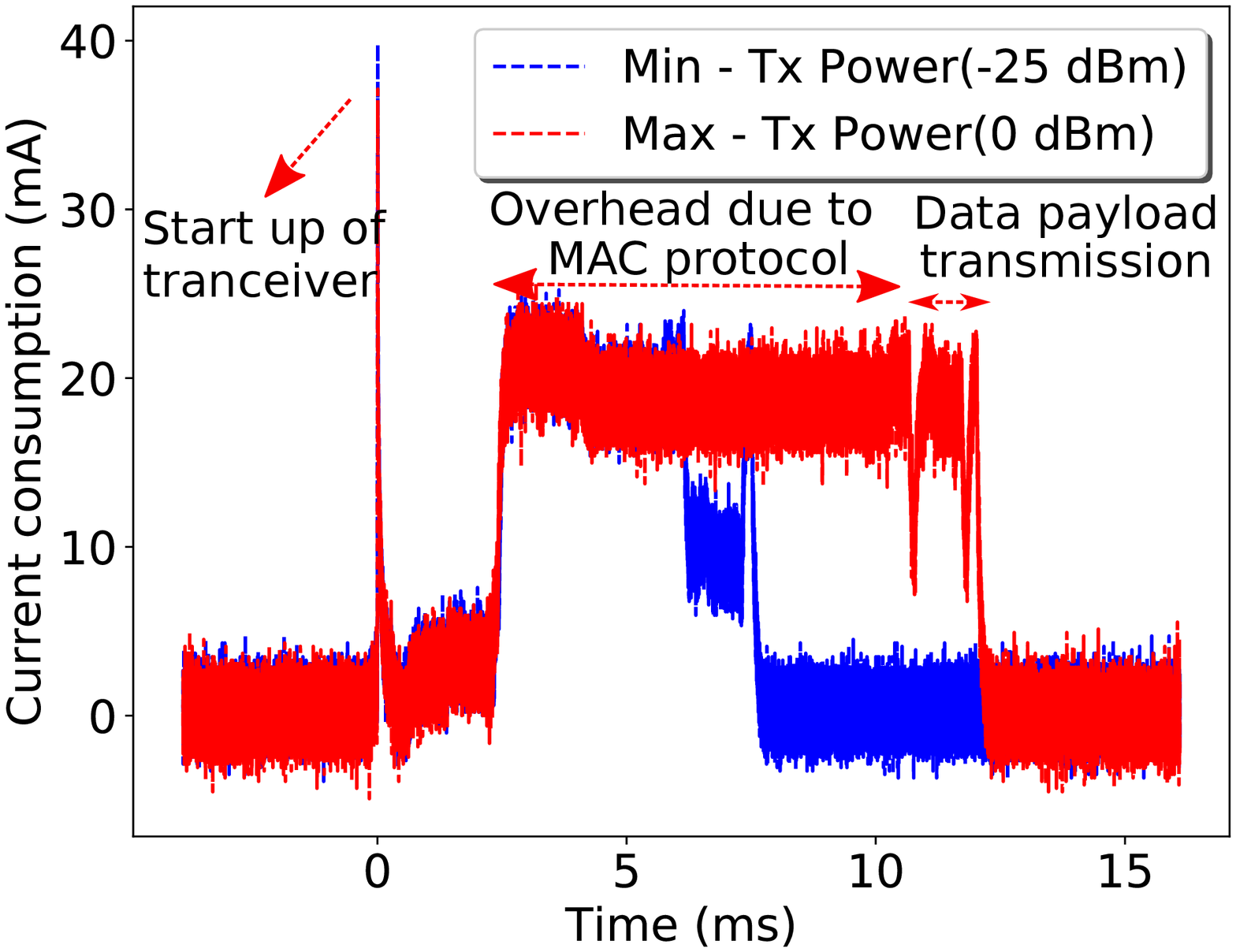}%
\label{label_minload}}
\caption{TelosB's current consumption at different payload sizes and transmission power levels}
\label{current_consumption_wrtpayload}
\end{figure}

To capture the energy consumption of a node over a long term for different data payload lengths, we programmed 3 TelosB child motes ($c_1, c_2, c_3$)
which periodically transmit one data packet (once in every 0.5 seconds) to their parents ($p_1, p_2, p_3$) as shown in Fig.~\ref{effect_payload}. 
During each wake-up period, each child node turns ON its radio, sends one data packet, turns the radio OFF and goes to the sleep stage. $c_1$, $c_2$ and $c_3$ send data packets with payload sizes of $10$ bytes, $50$ bytes and $100$ bytes, respectively. 
The reduction in the battery voltage levels as well as the remaining battery energy level for each node after periodically transmitting  data 
for 10 hours (Fig.~\ref{effect_payload}) reveal that \emph{all the three child nodes have almost the same amount of energy consumption
irrespective of the payload length}. This is because as shown in Fig.~\ref{current_consumption_wrtpayload}, the current consumption for the data transmission part
is only for a short duration (approximately 5 ms for a 114 byte data payload and 2 ms for a 10 byte data payload) in comparison with the time
required for a complete data packet transmission. The major portion of the current consumption (and hence energy) comes from various other activities like
radio startup, medium access control (MAC) etc.

\begin{figure}[h]
\centering
\subfloat[Effect of payload size]{\includegraphics[scale=0.75]{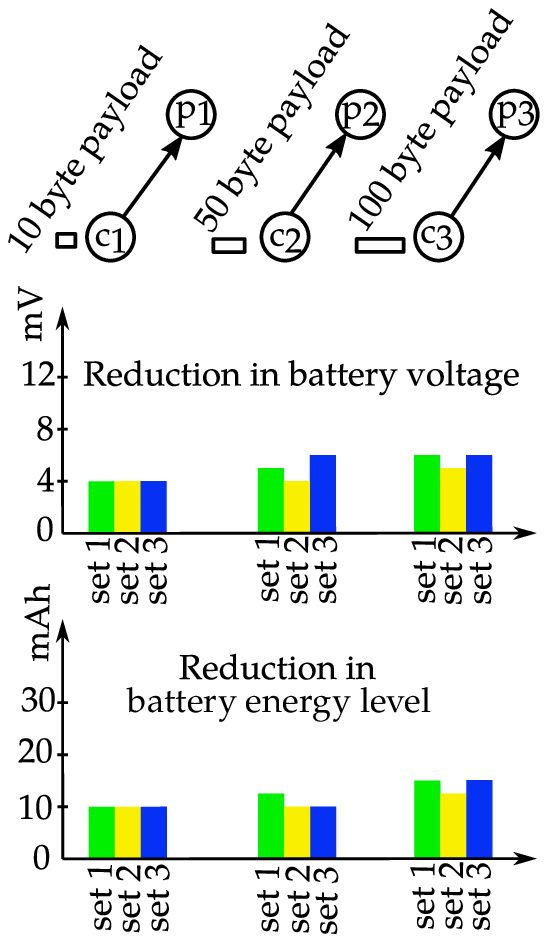}%
\label{effect_payload}}
\hfil
\subfloat[Effect of number of message transfers]{\includegraphics[scale=0.75]{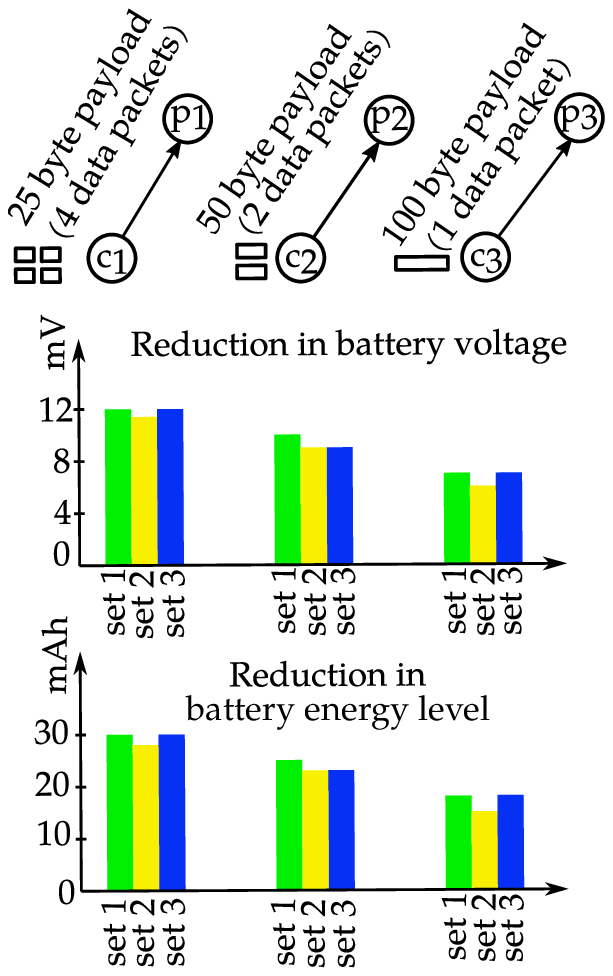}%
\label{effect_nooftransmissions}}
\caption{Reduction in battery voltage and remaining battery energy level}
\label{BV_RC_wrtpayload}
\end{figure}

To capture the energy consumption by a node for different numbers of packet transmissions, we programmed 3 TelosB child motes 
($c_1, c_2, c_3$) which periodically transmit a set of data packets to their 
parents ($p_1, p_2, p_3$) as shown in Fig.~\ref{effect_nooftransmissions}. In each wake-up period (once in every 0.5 seconds),
$c_1$ transmits $4$ data packets of payload length $25$ bytes each, $c_2$ transmits $2$ data packets of payload length $50$ bytes each,
and $c_3$ transmits $1$ data packet of payload length $100$ bytes.
Fig.~\ref{effect_nooftransmissions} reveals that \emph{the three child nodes differ substantially in the amount of reduction in the 
battery voltage levels as well as in the remaining battery energy level} after  periodically transmitting data for 10 hours.

The experiments corresponding to Fig.~\ref{effect_payload} and Fig.~\ref{effect_nooftransmissions}  were performed three times and the same trends were observed in each set. From the above detailed experiments, we observe that it is the \emph{number of data packet transmissions}, and \emph{not the number of bytes of data transmitted},  
which really creates energy imbalances among the nodes. If two nodes transmit the same number of packets, then the amount of energy consumption of the two nodes is approximately the same even if the payload lengths of the 
transmitted data packets are different. This is because the amount of energy required for the transmission of a data packet is approximately independent of the size of the data payload in the packet (see Fig.~\ref{effect_payload}).
There are several works in the literature which address the lifetime maximization problem for the convergecast operation in sensor networks 
\cite{wu2010constructing,junbin2010efficient,imon2015energy}.
These works consider the \emph{number of data units transmitted}, where a data unit is the amount of sensor data generated by a node in a time slot, as the parameter which determines the energy consumption of a node in their formulation of the lifetime maximization problem. However, the above experiments suggest that, instead, the \emph{number of data packets transmitted} needs to be considered as the parameter that determines the energy consumption of a node in the formulation. Hence, we formulate and solve the problem by considering the energy expenditure of a node in terms of the number of 
data packets it transmits (see the term $n_u^t(T)$ in (\ref{optimaleqn_final})). Note that in the piggyback aggregation model (PAM) (see Fig.~\ref{datacollection_schemes_wsn}), 
multiple data units can be concatenated to form a single data packet.
Also, from Fig.~\ref{BV_RC_wrtpayload}, it is clear that the PAM, which we use, is much more efficient than the packet relay model (PRM).

\subsection{Sensor energy consumption $k_u$}
During each wake-up period, a node mainly spends energy for three different activities: receiving data packets from its children, generating its own sensor data
and transmitting data packets to its parent. Most of the works in the literature ignore 
the energy consumed by a node for generating data from its own sensors since it is very small
in their applications. However, there can be scenarios where the sensor energy consumption cannot be ignored. 
Table~\ref{lifetable} gives an example of the energy consumed 
by different modules of a sensor node which was used in one of our earlier sensor network implementations targeted for an agricultural monitoring application~\cite{jobishjohnEAI}. 
The nodes were equipped with multiple sensors to measure soil temperature, soil moisture, atmospheric temperature and relative humidity.
The term $k_u$ in (\ref{optimaleqn_final}) captures the energy consumption of a sensor node due to activities like sensing and processing and Table~\ref{lifetable}
shows that it is not negligible. 

\begin{table}
\caption{Energy consumption of various sensor node elements: an example}
\label{lifetable}
\centering
\begin{tabular}{|c|c|c|c|}
\hline
\multirow{2}{*}{\bf{Element}}& \bf{Current} & \bf{Time} & \bf{Energy}\\
& \bf{(mA)}&\bf{(ms)}& \bf{(mJ)} \\
\hline
TelosB during transmission & 19.5 & 20 & 1.287 \\
TelosB during reception & 22 & 20 & 1.452 \\
Soil moisture sensor module & 80 & 5 & 1.32 \\
Other sensors & negligible & - & - \\
\hline
\end{tabular}
\end{table}

\subsection{Adjusting transmission power levels of the radio}
\label{adjust_txpw_level}
Recall that for a node $u, c_{e_{u-p_u{(T)}}} $ represents the energy cost for the transmission of one data packet to its parent $p_u{(T)}$.
In most wireless radios the transmission power of the radio is adjustable and hence, $c_{e_{u-p_u{(T)}}} $ is a variable parameter.
We use TelosB modules~\cite{telosb_mote} which use CC2420~\cite{cc2420datasheet} as the wireless radio. 
In one of our earlier works~\cite{john2015design}, we have studied the effect of transmission power on the wireless communication range.
For CC2420, when we increase the transmission power from the lowest level (-25 dBm) to the highest level (0 dBm), the communication range increases from 8.5 m to 56.5 m in an outdoor
non-line of sight scenario and the current consumption of the wireless module almost doubles (10 mA to 18.33 mA)~\cite{john2015design}. 
Hence, this is an important factor which we take into account in our formulation of the maximum lifetime tree construction problem. In our network model, each node sets its transmission power to the minimum value for which its data packet reaches the 
intended neighbour node. This serves two purposes--  it helps to save energy as well as reduces the interference caused 
to other nodes in the network.
In our formulation, the \emph{variable parameter} $c_{e_{u-p_u{(T)}}}$ in (\ref{optimaleqn_final}) represents the energy cost corresponding to the minimum transmission power required to send a packet from node $u$ to its parent $p_u{(T)}$.  On the other hand, in most prior works on the maximum lifetime tree construction problem, it is assumed that the radio transmission power is fixed and a \emph{constant} term is used instead of the term $c_{e_{u-p_u{(T)}}}$ in the formulation. 

\section{Complexity and Algorithm}
\label{hardness_solution}
\subsection{Complexity}
\begin{theorem}
 The maximum lifetime data collection tree problem in Problem~\ref{pblm1} is NP-complete.
\end{theorem}

\begin{proof}
 
The NP-completeness of Problem~\ref{pblm1} is proved by reducing the set-cover problem~\cite{kleinberg2006algorithm} to it. 

The decision version of Problem~\ref{pblm1} can be stated as: given a graph representing a WSN and a number $\tau$, does there exist a tree 
with lifetime at least $\tau$? If we are 
given a solution (tree), we can easily verify in polynomial time whether its lifetime is $\ge$ $\tau$.
This proves that Problem~\ref{pblm1} is in class NP~\cite{kleinberg2006algorithm}.

Next, we reduce the set cover problem, which is known to be NP-complete~\cite{kleinberg2006algorithm}, to Problem~\ref{pblm1}. An instance of the set cover problem is as follows. Assume that there are $n$ elements, $1,2,3,...,n$, in the set $U$. Let $B_1,B_2...B_k$ be $k$ given subsets of $U$ and $p$ be a given number. 
The set cover problem is to find whether there exists any collection of at most $p$ of these subsets
 whose union is $U$. We refer to a collection of $p$ subsets whose union is $U$ as a ``$p$ set cover''.

We reduce the above instance of the set cover problem to the special case of Problem~\ref{pblm1} in which $e_r = 0$, $l=\beta,$
$k_u = a \ \forall u $, and
$c_{e_{u-v}} = 1 \ \forall e_{u-v} \in E$. $\beta$ represents the maximum number of data bytes in one data packet; so $l=\beta$ implies that each data packet can contain only one data unit.
For the reduction, we construct a connectivity graph as shown in Fig.~\ref{setcoverredution}.

\begin{figure}[h]
 \centering
\includegraphics[scale=0.6]{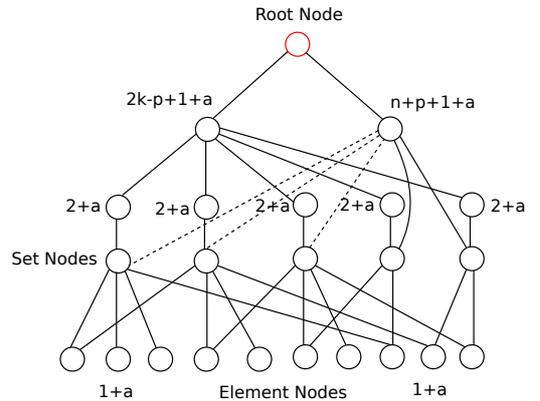}
 \caption{Reduction from set cover to Problem~\ref{pblm1}}
 \label{setcoverredution}
\end{figure}

The root node (sink) has infinite energy and is connected to two nodes having energies, $2k-p+1+a$ and $n+p+1+a$. The third row consists of $k$ nodes each with energy
$2+a$ and each node in this row is connected to the node with energy $2k-p+1+a$. The fourth row consists of $k$ nodes, each representing one of 
the given subsets, and the $i$'th node in this row has energy $|B_i|+1+a$.
All the nodes of the fourth row have a connection to the node with energy $n+p+1+a$. The last row contains $n$ nodes, each node corresponding to one of  the $n$ elements in the set $U$, and each of these nodes has connections to the nodes in the fourth row representing the subsets the corresponding element is contained in.

We will show that the graph in Fig.~\ref{setcoverredution} has a lifetime of one if and only if there exists a $p$ set cover in the above instance of the set cover problem.  

Assume that there exists a $p$ set cover. The tree with lifetime one can be constructed  by the following steps.
\begin{itemize}
% \item {Connect all those p subset equivalent node from fourth row to the node with energy $n+p+1+a$.}
 \item {Connect the $p$ nodes in the fourth row representing the $p$ subsets of the $p$ set cover to the node with energy $n+p+1+a$ (represented by dotted edges in Fig.~\ref{setcoverredution}).}
 \item {Ensure that each element node in the last row has a connection to one of the $p$ subset nodes.}
 \item {The remaining $(k-p)$ subset representing nodes from the fourth row are connected to the node with energy $2k-p+1+a$ through their corresponding nodes in the third row.}
 \item {Also, ensure that all the nodes in the third row have a connection to the node with energy $2k-p+1+a$.}
 \item {Connect the two nodes in the second row with the root.}
 \end{itemize}
 It is easy to check that the above connections constitute a tree with lifetime one.

Now assume that there is a tree with lifetime one. Since each node in the third row only has energy $2+a$, it can route data from only one 
node corresponding to a subset from the fourth row to the root node and the node with energy $2k-p+1+a$ limits the routing to $(k-p)$ set equivalent nodes.
So, the $p$ nodes remaining in the fourth row have to cover all the $n$ elements in the last row. This implies that there exists a $p$ set cover. 

The result follows.
\end{proof}

An approach similar to that used in the above proof is used in~\cite{buragohain2005power}.

\subsection{Proposed tree construction algorithm}
Here we propose an algorithm for finding a data collection tree, called Balanced energy consumption Data Collection Tree (BDCT), 
which provides a high network lifetime. 
The tree construction algorithm is outlined in Algorithm~\ref{alg:algorithm_treeconstruction}.
The algorithm is executed by the sink node which has complete information about the network.
As discussed in Section~\ref{Network_architecture}, the neighbour details of a node consist of its neighbour 
node ids along with the minimum transmission energy required for reliable communication with them. 
The tree building starts with the sink node and during each iteration an edge is added to the partially built tree (subtree) which connects 
a node which is not yet covered in the current subtree. 
Let $S_T$ be a set of paired values where each pair in $S_T$ represents a node id of a node that is already added to the data collection tree,
along with its hop distance from the sink node.
 For example, if the pair $(u,h_u)$ belongs to  $S_T$, then node $u$ is $h_u$ hops away from the sink node
in the current data collection subtree.
Let $E_T$ be the set of possible edges which can be added next, to connect a node which is not yet connected, to the current subtree. 

\begin{algorithm}
  \caption{Algorithm for BDCT construction}
 \begin{algorithmic}[1]
 \renewcommand{\algorithmicrequire}{\textbf{Input:}}
 \renewcommand{\algorithmicensure}{\textbf{Output:}}
 \REQUIRE $G=(V,E)$, $c_{e_{u-v}} \forall e_{u-v} \in E$, $B_u \forall u \in V$  
 \ENSURE  A data collection tree, $T$
 \\ \textit{Initialisation} : $S_T = \{\}, E_T = \{\}, f_{E_T}(e_{u-v}) = \{ \}$
  \STATE Add sink node, $(s,0)$, to $S_T$
  %\STATE Add $e_{u-s}$ to $E_T$, $\forall e_{u-s} \in E$
  \STATE Add $e_{u-s}$ to $E_T$, $\forall u \in V$ such that $e_{u-s} \in E$
  \WHILE {(any uncovered node exists)}
   \STATE For each edge, say $e_{u-v}$, in $E_T$, assign a real number $f_{E_T}(e_{u-v})$ using (\ref{mapping_e_n})
   \STATE Select next edge, say $e_{next}$, to be added to the subtree using (\ref{shortlist_eqn})
   \STATE Update the subtree, $S_T$, and $E_T$
  \ENDWHILE
 \RETURN $T$ 
 \end{algorithmic} 
 \label{alg:algorithm_treeconstruction}
 \end{algorithm}

Initially, $S_T = \{ \ \},  E_T = \{ \ \}$.
As a first step, the algorithm adds the sink node $s$ to $S_T$ and thus $S_T = \{ (s,0) \}$. Now, $E_T$ will contain all the edges 
which connect the sink node $s$ to any other node in the connectivity graph. During each iteration, the algorithm uses a mapping $f_{E_T}(e_{u-v})$ from each $e_{u-v} \in E_T$ given by (\ref{mapping_e_n}), where $u$ is an uncovered node and $v$ is a covered node, and using this mapping,  selects an 
edge $e_{next}$ from $E_T$ (see (\ref{shortlist_eqn})). The edge $e_{next}$ is added to the current data collection subtree. 
 
\begin{multline}
\label{mapping_e_n}
f_{E_T}(e_{u-v}) = \min \Bigg(\dfrac{B_u}{c_{e_{u-v}} + h_v . C_{avg} + k_u}, \\
\dfrac{B_v}{h_v.C_{avg}\ceil{\frac{(\alpha_v+2).l}{\beta}} + (\alpha_v + 1).e_r + k_v}\Bigg) \forall e_{u-v} \in E_T 
\end{multline}

\begin{equation}
\label{shortlist_eqn}
 e_{next} = \underset{e_{u-v} \in E_T}{\argmax} \{f_{E_T}(e_{u-v})\} 
\end{equation}
$C_{avg}$ is the average packet transmission 
energy (considering different possible transmission power levels) and $\alpha_v$ is the current load (number of children connected) of node~$v$. 
If multiple edges achieve the maximum in the term $\underset{e_{u-v} \in E_T}{\argmax} \{f_{E_T}(e_{u-v})\}$  in (\ref{shortlist_eqn}), then the edge $e_{u-v}$ with the
highest non-minimizing term in the $\min$ in (\ref{mapping_e_n}) is selected as $e_{next}$. The mapping $f_{E_T}(e_{u-v})$ outputs a real number, say $n_{e_{u-v}}$, for each $e_{u-v}$ which is an indicative measure of how advisable it is, to connect
the uncovered node $u$ to a covered node $v$ as $u$'s parent using the edge $e_{u-v}$. 

The intuition behind the function $f_{E_T}(e_{u-v})$ is as follows: 
\begin{enumerate}
 \item {It takes a higher value if nodes $u$ and $v$ have higher battery voltages, $B_u$ and $B_v$, respectively. This means nodes with higher battery voltages are preferred to be added first to the data collection subtree and hence they are more likely to act as parent nodes for uncovered nodes.}
  \item {The term $ h_v . C_{avg}$ in (\ref{mapping_e_n}) is used to indirectly penalize long hop paths to the sink from other nodes.}
 \item {Addition of a child node $u$ to $v$ will increase the current load (number of child nodes) on $v \ (\alpha_v)$ by one and hence will increase 
 the energy expenditure of node $v$. The new energy expenditure of node $v$ in packet reception and transmission after the addition of the edge $e_{u-v}$ is addressed by the term  $h_v.C_{avg}\ceil{\frac{(\alpha_v+2).l}{\beta}} + (\alpha_v + 1).e_r$ in $f_{E_T}(e_{u-v})$.}
 \end{enumerate}

\begin{figure*}
\centering
\subfloat[Connectivity graph]{\label{BDCT_example_graph} \includegraphics[scale=0.6]{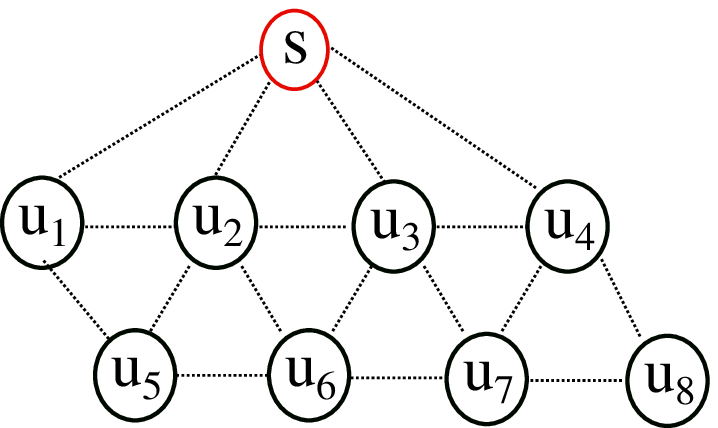}}
\hfil
\subfloat[Data collection subtree after the addition of edge $e_{u_1-u_2}$]{\label{heuristic_step3} \includegraphics[scale=0.6]{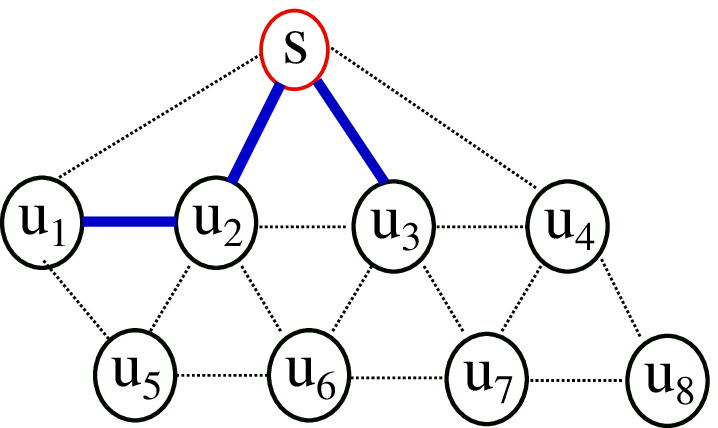}}
 \hfil
\subfloat[Data collection tree after the algorithm termination]{\label{heuristic_step5} \includegraphics[scale=0.6]{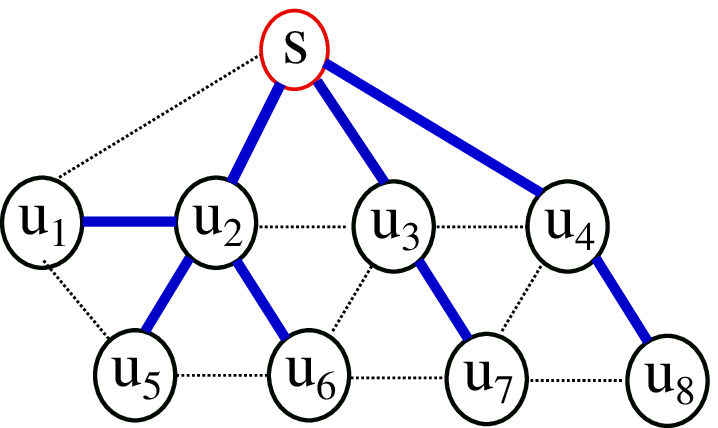}}
\hfil
\caption{An example, which illustrates the execution of the proposed BDCT construction algorithm}
\label{BDCT_example}
\end{figure*}

We illustrate the construction of a tree using the proposed algorithm using an example, which is shown in Fig.~\ref{BDCT_example}. 
Consider the network graph shown in Fig.~\ref{BDCT_example_graph}.
During the first step, addition of the sink node $s$ to $S_T$ results in $S_T = \{ (s,0) \}$,
$E_T = \{e_{u_1-s},e_{u_2-s},e_{u_3-s},e_{u_4-s} \} $ and the corresponding mapping 
$f_{E_T}(e_{u-v}) \in \{n_{e_{u_1-s}},n_{e_{u_2-s}},n_{e_{u_3-s}},n_{e_{u_4-s}} \} $. For the current $E_T$,  
\begin{multline*}
f_{E_T}(e_{u-v}) \in \bigg\{\min \left( \dfrac{B_{u_1}}{c_{e_{u_1-s}}+0+k_{u_1}}, \infty\right), \\
\min \left( \dfrac{B_{u_2}}{c_{e_{u_2-s}}+0+k_{u_2}}, \infty\right), \min \left( \dfrac{B_{u_3}}{c_{e_{u_3-s}}+0+k_{u_3}}, \infty\right), \\
 \min \left( \dfrac{B_{u_4}}{c_{e_{u_4-s}}+0+k_{u_4}}, \infty\right) \bigg\}
\end{multline*}

In the current iteration, let $e_{next}$ be $e_{u_2-s}$.
$u_2$ gets marked as a covered node and its hop distance from the sink node is updated by adding one hop to its parent's hop count ($h_{u_2} = h_s + 1 = 1$).
The load of the parent node is also incremented by one; $\alpha_s = \alpha_s + 1 =1$.

In the next iteration, suppose $e_{next} = e_{u_3-s}$ gets connected to the data collection subtree; after updation, we get $h_{u_3} = 1$ and $\alpha_s = 2$.
The updated $S_T = \{(s,0), (u_2,1), (u_3,1)\}$ and $E_T = \{e_{u_1-s},e_{u_4-s},e_{u_1-u_2},e_{u_5-u_2},e_{u_6-u_2},e_{u_6-u_3}, e_{u_7-u_3},e_{u_4-u_3} \}$.
The corresponding mapping 
\begin{align*}
f_{E_T}(e_{u-v}) & \in \Bigg\{\min \left( \dfrac{B_{u_1}}{c_{e_{u_1-s}}+0+k_{u_1}}, \infty\right),  \ ... \\
&  \min \Bigg( \dfrac{B_{u_1}}{c_{e_{u_1-u_2}}+h_{u_2}.C_{avg}+k_{u_1}},\\
&  \dfrac{B_{u_2}}{h_{u_2}.C_{avg}\ceil{\frac{(\alpha_{u_2}+2).l}{\beta}} + (\alpha_{u_2} + 1).e_r + k_{u_2}}\Bigg), \ ... \ \Bigg\}\\
\end{align*}

During this iteration, assume that all the nodes except $u_1$ have very high energy and therefore $e_{next}$ 
is decided by the terms $\dfrac{B_{u_1}}{c_{e_{u_1-s}}+0+k_{u_1}}$ and $\dfrac{B_{u_1}}{c_{e_{u_1-u_2}}+h_{u_2}.C_{avg}+k_{u_1}}$. 
The first term corresponds to connection of $u_1$ directly to $s$ while the second term corresponds to connection of $u_1$ to $s$ via $u_2$. 
Assume that the energy required for transmitting a data packet directly from $u_1$ to $s$ ($c_{e_{u_1-s}}$) is greater than that of 
the multi-hop path ($c_{e_{u_1-u_2}}+h_{u_2}.C_{avg}$); then, edge $e_{u_1-u_2}$ is added to the data collection subtree. The resulting subtree is shown in Fig.~\ref{heuristic_step3}.

Assume that during the next iteration, the edge $e_{u_4-s}$ gets added to the data collection subtree. 
The updated $S_T = \{(s,0), (u_2,1), (u_3,1), (u_1,2), (u_4,1)\}$ and
$E_T = \{e_{u_5-u_2},e_{u_6-u_2},e_{u_6-u_3},e_{u_7-u_3},e_{u_5-u_1},e_{u_7-u_4},e_{u_8-u_4} \}$. 
The corresponding mapping $ f_{E_T}(e_{u-v}) \in \{...\ , n_{e_{u_6-u_2}},n_{e_{u_6-u_3}},\ ...\ \}$. That is,

\begin{align*}
f_{E_T}(e_{u-v}) & \in \Bigg\{ ...,\  \min \Bigg( \dfrac{B_{u_6}}{c_{e_{u_6-u_2}}+h_{u_2}.C_{avg}+k_{u_6}},\\
&  \dfrac{B_{u_2}}{h_{u_2}.C_{avg}\ceil{\frac{(\alpha_{u_2}+2).l}{\beta}} + (\alpha_{u_2} + 1).e_r + k_{u_2}}\Bigg), \\
& \min \Bigg( \dfrac{B_{u_6}}{c_{e_{u_6-u_2}}+h_{u_2}.C_{avg}+k_{u_6}},\\
&  \dfrac{B_{u_3}}{h_{u_3}.C_{avg}\ceil{\frac{(\alpha_{u_3}+2).l}{\beta}} + (\alpha_{u_3} + 1).e_r + k_{u_3}}\Bigg),\ ... \Bigg\}\\
\end{align*}
During this iteration,  assume that two edges, $e_{u_6-u_2}$ and $e_{u_6-u_3}$, are both maximizers in the term ${\max} \{f_{E_T}(e_{u-v})\}$; so $n_{e_{u_6-u_2}}$ and $n_{e_{u_6-u_3}}$ are equal.
Hence, in this iteration there are two edges which can be added next to the data collection subtree. In such iterations, we propose to add the edge which has the higher second term in the $\min$ in (\ref{mapping_e_n}) among the short-listed edges ($e_{u_6-u_2}$ and $e_{u_6-u_3}$). This indirectly 
tries to connect a node to a parent which has higher lifetime. Thus our algorithm selects $e_{next} = e_{u_6-u_2}$ if:
\begin{multline*}
 \dfrac{B_{u_2}}{h_{u_2}.C_{avg}\ceil{\frac{(\alpha_{u_2}+2).l}{\beta}} + (\alpha_{u_2} + 1).e_r + k_{u_2}} > \\
\dfrac{B_{u_3}}{h_{u_3}.C_{avg}\ceil{\frac{(\alpha_{u_3}+2).l}{\beta}} + (\alpha_{u_3} + 1).e_r + k_{u_3}}
\end{multline*}
In this way the algorithm proceeds and terminates once all the nodes in the network are covered or added to the data collection subtree. The constructed 
data collection tree for the given example is shown in Fig.~\ref{heuristic_step5}.

\section {Performance Evaluation}
\label{Performance_evaluation}
First, in Section~\ref{SSC:testbed:based:evaluation}, we evaluate the  performance of the proposed algorithm via its actual implementation on a WSN testbed.  Then, in Section~\ref{SSC:simulations}, we present a performance evaluation of the proposed algorithm in large networks through simulations. 

\subsection{Testbed based performance evaluation of proposed algorithm}
\label{SSC:testbed:based:evaluation}
\subsubsection{Experimental procedure}
The performance of the  proposed lifetime maximization algorithm was evaluated through its actual implementation on a WSN testbed. The testbed
consists of 20 sensor nodes which are installed in an indoor environment as shown in Fig.~\ref{placement_nodes}. Each sensor node consists of a TelosB mote and is powered using a 
Li-Ion battery of capacity 2200 mAh. We use the piggyback aggregation model (PAM) (see Section~\ref{Network_architecture}) for data collection. 

\begin{figure}
\centering
\subfloat[Locations of sensor nodes in the testbed]{\includegraphics[scale=0.38]{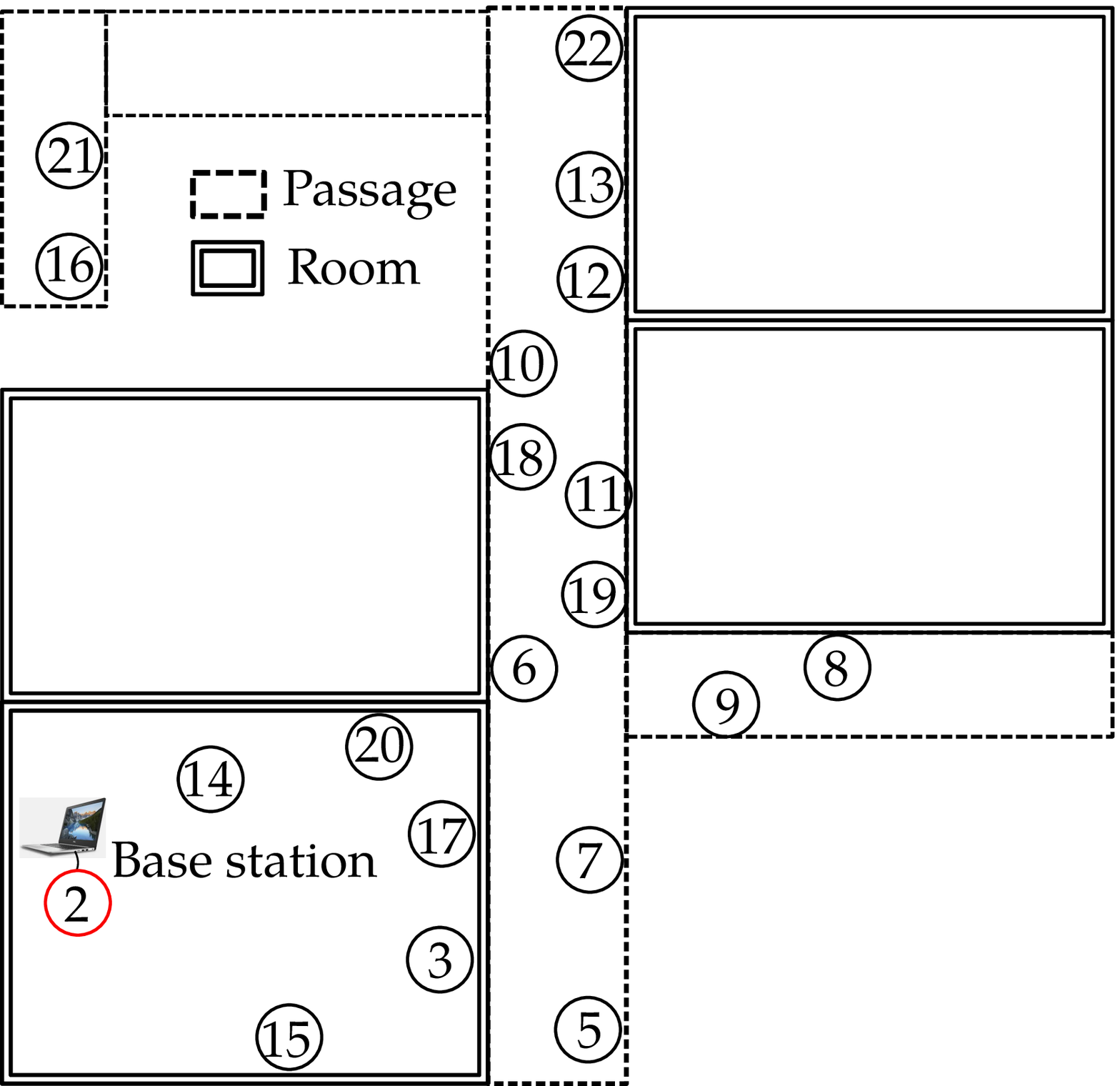}}
\hfil
\subfloat[Sensor node]{\includegraphics[scale=0.45]{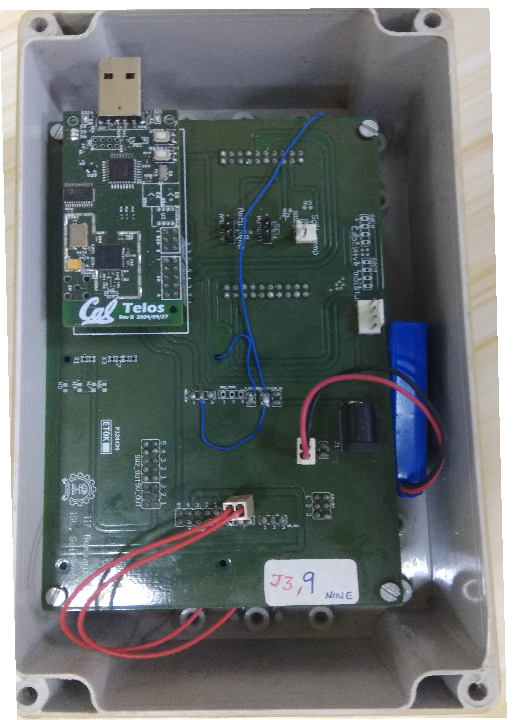}}
 \hfil
\caption{The wireless sensor network testbed}
\label{placement_nodes}
\end{figure}

In this section, the performance of the proposed lifetime maximization approach, BDCT, is compared with those of the shortest path tree (SPT) algorithm and random data collection tree (RDCT) algorithm. SPT is constructed by running Dijkstra's algorithm~\cite{Dijkstra} at the sink node once it has collected 
the network connectivity graph with weights assigned to the edges in the same manner as detailed in Section~\ref{Network_architecture}. The RDCT construction approach starts 
with the sink node and randomly adds an uncovered node as a child in each step to the partially built data collection tree until all the nodes are covered.

Different operations occur in the network deployed in the testbed as explained in Section~\ref{Network_architecture} and are shown in Fig.~\ref{periodic_operations_toplevel}.
The network under consideration is a homogeneous sensor network. That is, each node in the network is equipped with the same types of sensors and 
hence different nodes consume constant and equal amounts of energies for generating data from their 
sensors in each wake-up period (represented by $k_u$ for node $u$ in (\ref{optimaleqn_final})). Considering these facts, the difference in the energy consumption of different nodes mainly arises from 
the periodic data collection phase because of different numbers of descendants that different nodes have
in the data collection tree. A node with a higher number of descendants is expected to spend more energy and vice-versa.

The reduction in the battery voltage of a node due to a single packet transmission (even with maximum payload length) is very small.
In order to have a noticeable battery voltage reduction, there have to be several packet transmissions by the node. So for the lifetime 
performances of different algorithms (BDCT, SPT and RDCT) to be noticeably different, a large number of packets must be transmitted by nodes 
using  the trees constructed by these algorithms.  Hence, the performances of different algorithms can be compared in two ways. 

The first approach is to perform the convergecast operation along the data collection tree constructed by each algorithm for a large number of wake-up periods
so that there is sufficiently enough discharge in the battery voltage of each node in the network. This approach has several practical 
difficulties. Even though the nodes in the network are synchronized during the data collection tree formation stage, all nodes are not perfectly 
synchronized to each other. Because of synchronization errors, some nodes may spend more time in the active stage and drain more energy in each
wake-up period than others. The synchronization errors between nodes grow with each wake-up period because of clock drifts~\cite{karl2007protocols} and it calls 
for the execution of a periodic synchronization strategy, which will be another energy overhead. Note that the energy consumption in synchronization are not the same for all the nodes in the network, which may lead to an unfair comparison of the lifetime performances of different algorithms.

The second approach for the performance comparison, which we use, is to consider that each node has a very large amount of data which needs to be sent to the sink node in each wake-up period. So even when data collection happens only for one wake-up period, there is sufficiently noticeable discharge in the battery voltages. Thus, in our evaluation strategy, we considered each node to have $30,000$ generated data packets ($100$ bytes each) to transfer to 
its parent node in one wake-up period, in addition to the data packets received from its descendants. For example, a node with two descendants has to send a total of $90,000$ packets to its parent node.

Once all the nodes in the network are turned ON, they enter into the neighbour discovery phase and the sink node constructs the
data collection tree as detailed in Section~\ref{Network_architecture} using a particular tree construction approach (BDCT or SPT or RDCT).
Once the data collection tree is built, during the wake-up period, each node receives data packets from its children, sends them along with its 
generated data packets to its parent node and 
goes to the sleep stage. Thus, data is collected at the sink node from all the nodes in the network.
We call this process as one ``data collection round''. This process is performed for seven continuous data collection rounds, which 
constitutes one test case. Each node in the testbed is equipped with an approximately fully charged lithium ion battery at the beginning of a test case.
Even though the batteries are almost fully charged, the initial voltages of the batteries are not exactly the same. They are observed to be within the range
4.18 - 4.2 volts. The battery voltage of each node is measured before and after each data collection round.
All the activities occurring in one test case as well as in one data collection round are shown in Fig.~\ref{testcase_details}.
Three test case measurements are carried out for each tree construction algorithm (BDCT, SPT and RDCT).

\begin{figure}[h]
\centering
\includegraphics[scale=0.4]{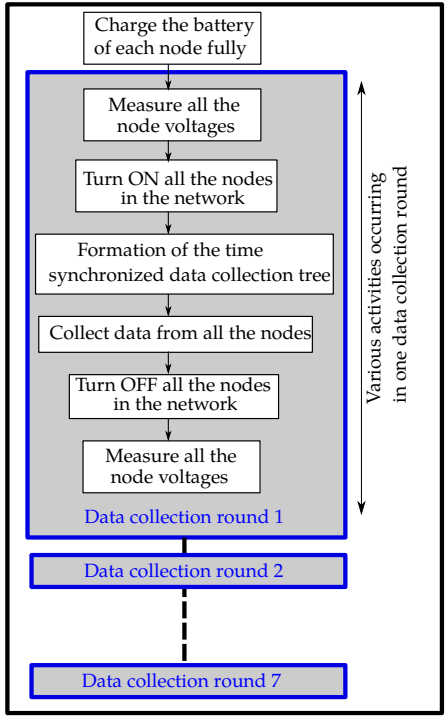}
\caption{Activities occurring in one test case}
\label{testcase_details}
\end{figure}

Now, a fairness index~\cite{jain1984quantitative} is a quantitative measure for the degree of fairness of resource allocation in a distributed system.
\emph{Jain's fairness index}~\cite{jain1984quantitative} is a measure of fairness of the resources allocated to $n$ users, where user $i$ receives value  (amount of resource) $x_i$, and is defined as: 
\begin{equation}
 J(x_1,x_2,x_3,...,x_n) = \dfrac{(\sum_{i=1}^{n} x_i)^2}{n \times \sum_{i=1}^{n} x_i^2}
\end{equation}
Jain's fairness index lies between 0 and 1.
We use the Jain's fairness index on the battery voltages after each data collection round as a measure to identify how close to each other the battery voltages are. 
If all the nodes have the same battery voltage, $J$ takes the value $1$. If $p$ nodes out of $n$ nodes have equal battery voltages and the remaining
nodes have zero battery voltage, then $J$ takes the value $\dfrac{p}{n}$. 
When the battery voltages are close to each other, the Jain's fairness index is close to one and vice-versa. Various other properties of the Jain's fairness index are detailed in~\cite{jain1984quantitative} and~\cite{chiu1989analysis}.

\subsubsection{Experimental Results}

\begin{figure}[h]
\centering
\includegraphics[scale=0.3]{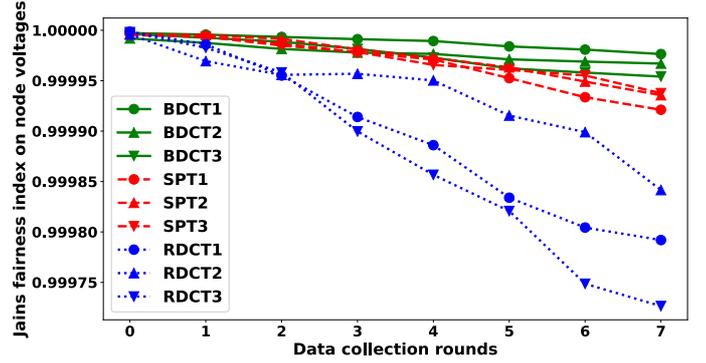}
\caption{Comparison of the Jain's fairness indices under the proposed approach, shortest path tree and random tree approaches}
\label{performace_propo_djk}
\end{figure}

Fig.~\ref{performace_propo_djk} compares the data collection performance of the proposed approach with those of the shortest path tree and random data collection tree approaches.
BDCT1, BDCT2, BDCT3 are three test cases where the proposed lifetime maximization  
algorithm is used for building the data collection tree. SPT1, SPT2 and SPT3 (respectively, RDCT1, RDCT2 and RDCT3) are three test cases where the sink node builds the shortest path tree using 
Dijkstra's algorithm (respectively, builds a random tree).
The major observations from Fig.~\ref{performace_propo_djk} are as follows:

\begin{enumerate}
 \item {In each test case, the Jain's fairness index at the start ($0^{th}$ data collection round)  is close to one, which indicates that the
 initial battery voltages for different nodes are close to each other. However, the Jain's fairness index is not exactly one,  which indicates that all nodes do not have exactly the same voltage to start with even though they are fully charged. Also, this value is slightly different for different test cases as well.}
  \item {The Jain's fairness index decreases as the data collection round number increases for all the three  data collection tree construction algorithms. Intuitively, this is due to the following reason.
In the WSN, some nodes are far from the sink and others are close to it.  Correspondingly, each tree has leaf nodes as well as nodes close to the root and hence there is a difference in the energy consumption of different nodes in a 
 particular data collection round, leading to unequal battery voltage levels. This results in the reduction of Jain's fairness index with the data collection round number.}
 \item {The \emph{rate of decrease in the Jain's fairness index is least for BDCT and most for RDCT}. This indicates that 
 the \emph{proposed BDCT collects data in a more load balanced manner compared to SPT and RDCT}. The proposed approach utilizes the nodes with higher battery energy to relay more data packets. Thus BDCT provides a longer lifetime for the network compared to SPT and RDCT. The same trend is consistently observed 
in all the test cases studied.}
 \end{enumerate}
 
 \begin{figure*}[h]
\centering
\subfloat[BDCT - test case 1]{\includegraphics[scale=0.21]{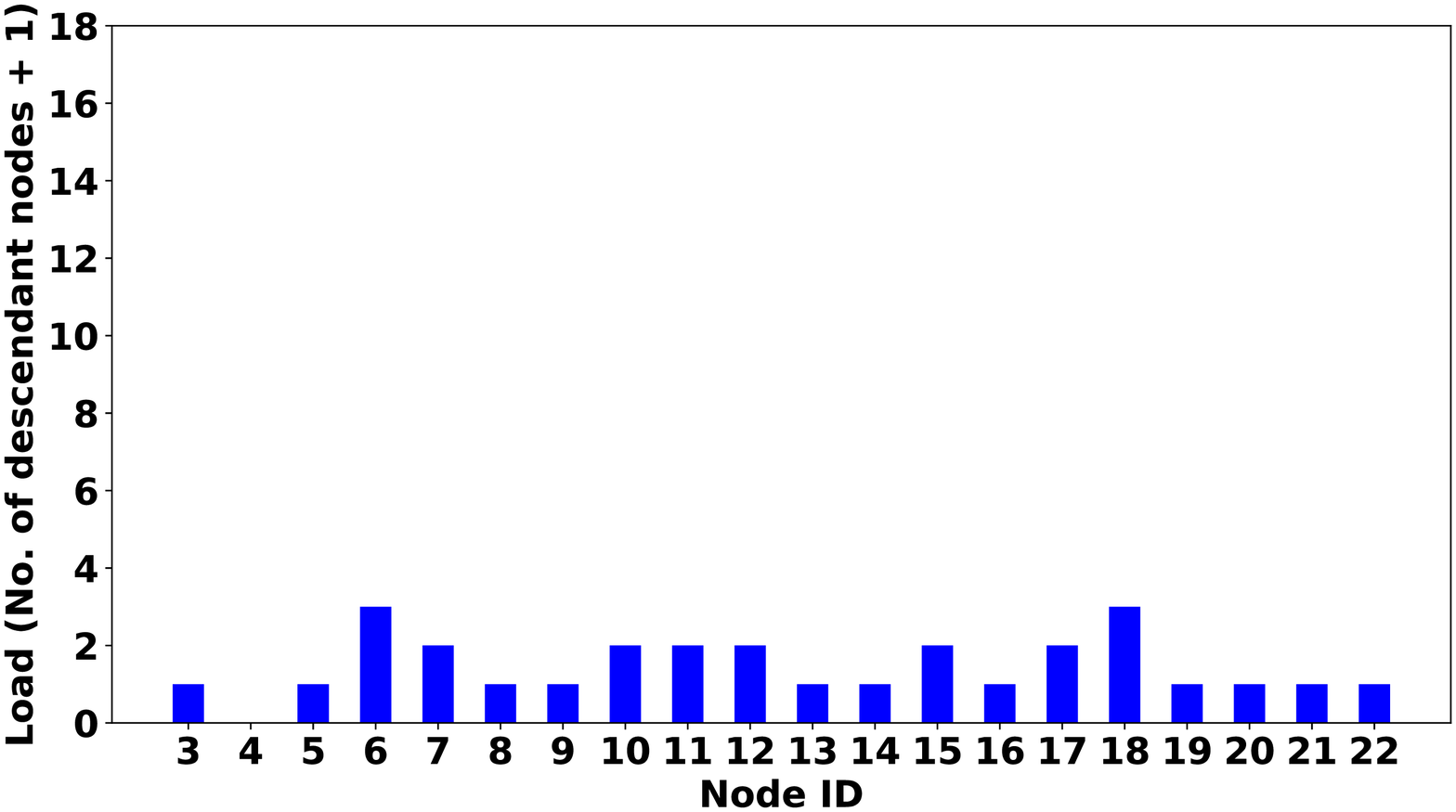}}
\hfil
\subfloat[BDCT - test case 2]{\includegraphics[scale=0.21]{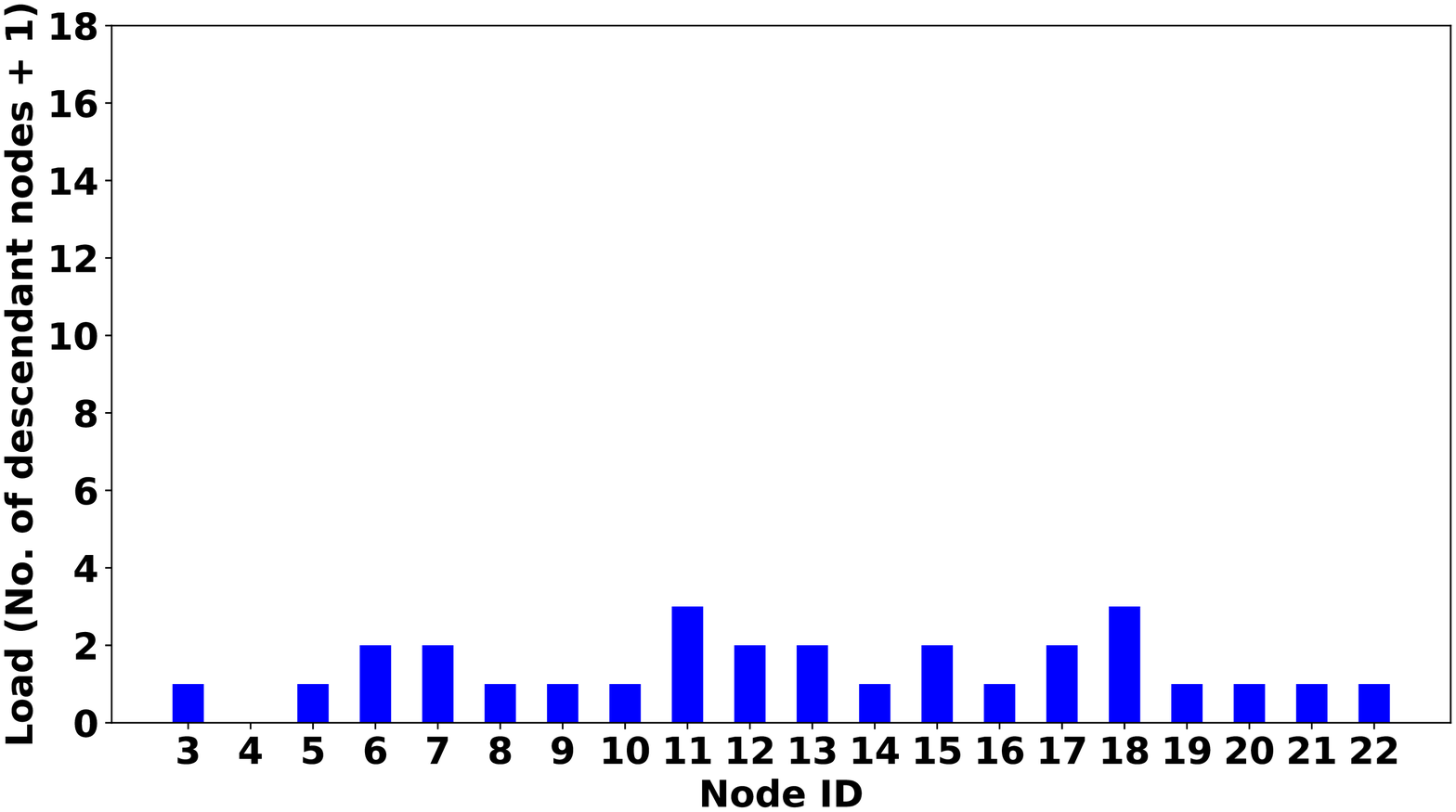}}
\hfil
\subfloat[BDCT - test case 3]{\includegraphics[scale=0.21]{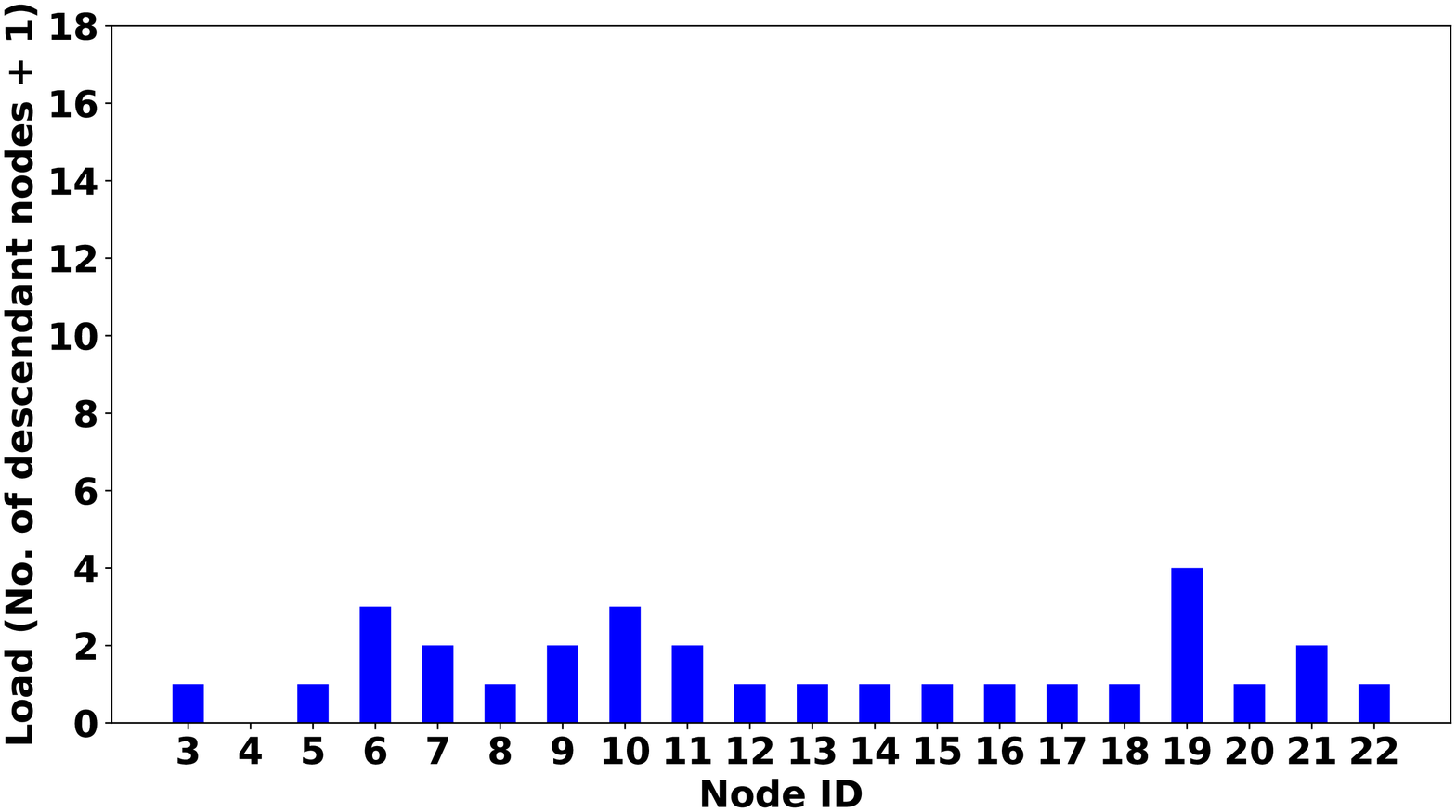}}
\hfil
\subfloat[SPT - test case 1]{\includegraphics[scale=0.21]{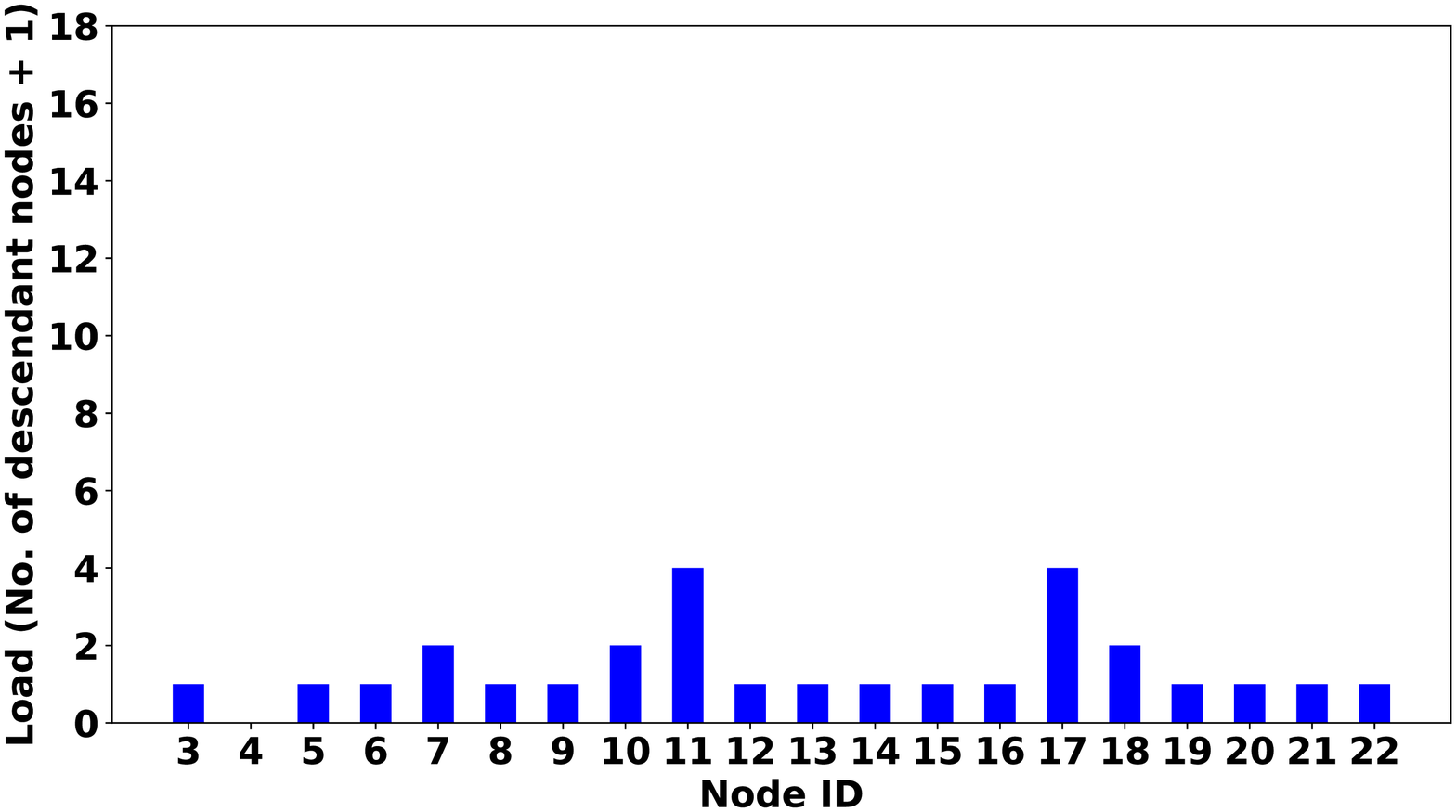}}
\hfil
\subfloat[SPT - test case 2]{\includegraphics[scale=0.21]{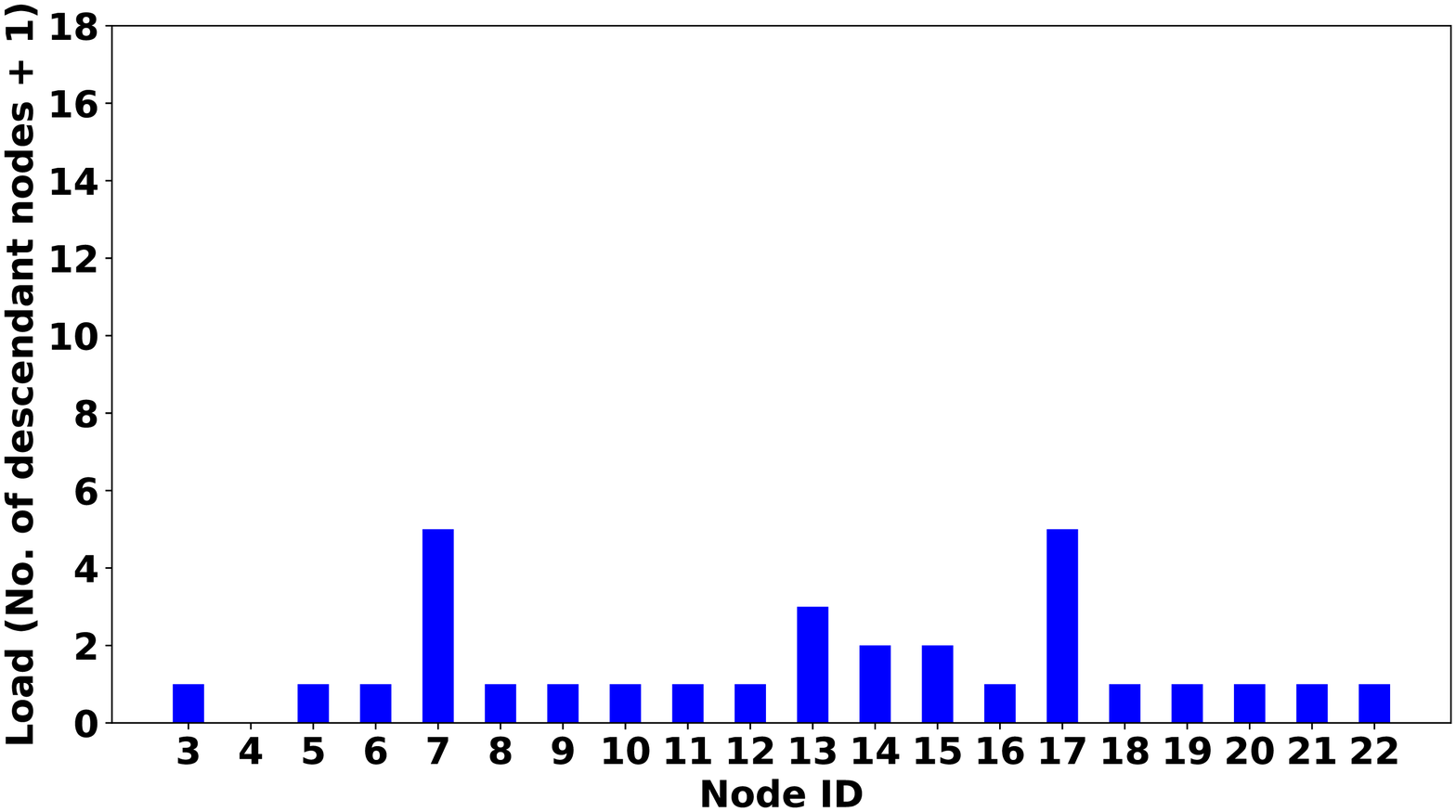}}
\hfil
\subfloat[SPT - test case 3]{\includegraphics[scale=0.21]{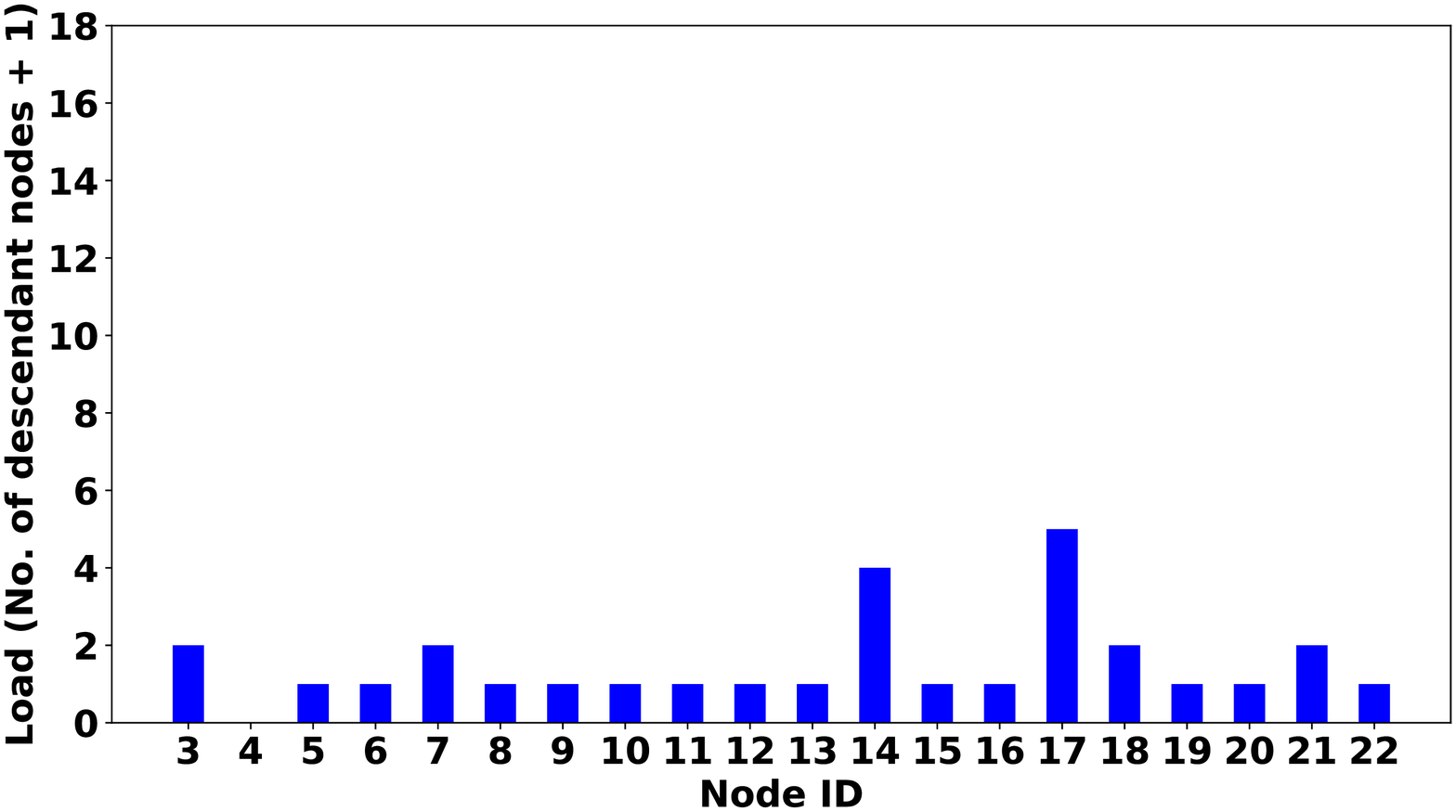}}
\hfil
\subfloat[RDCT - test case 1]{\includegraphics[scale=0.21]{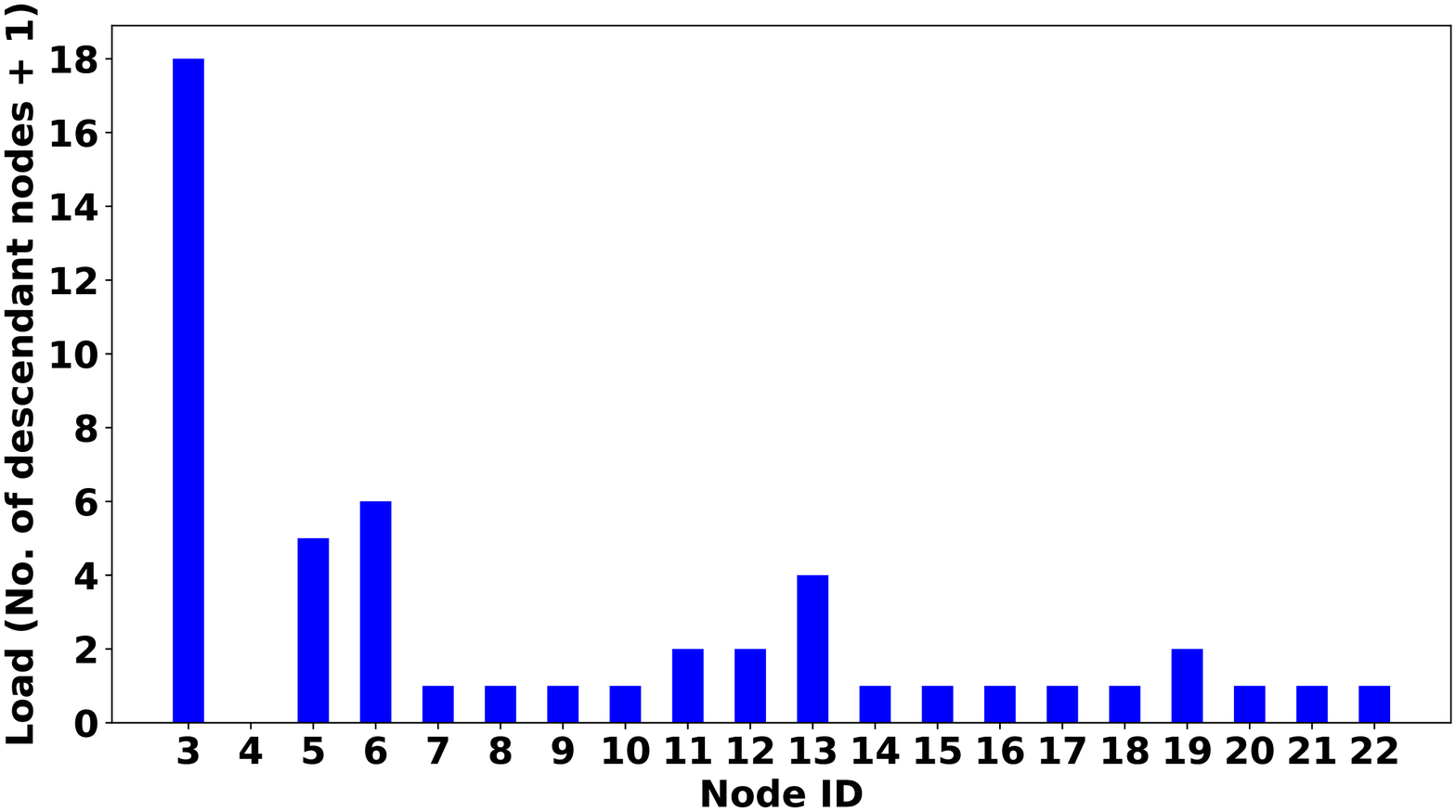}}
\hfil
\subfloat[RDCT - test case 2]{\includegraphics[scale=0.21]{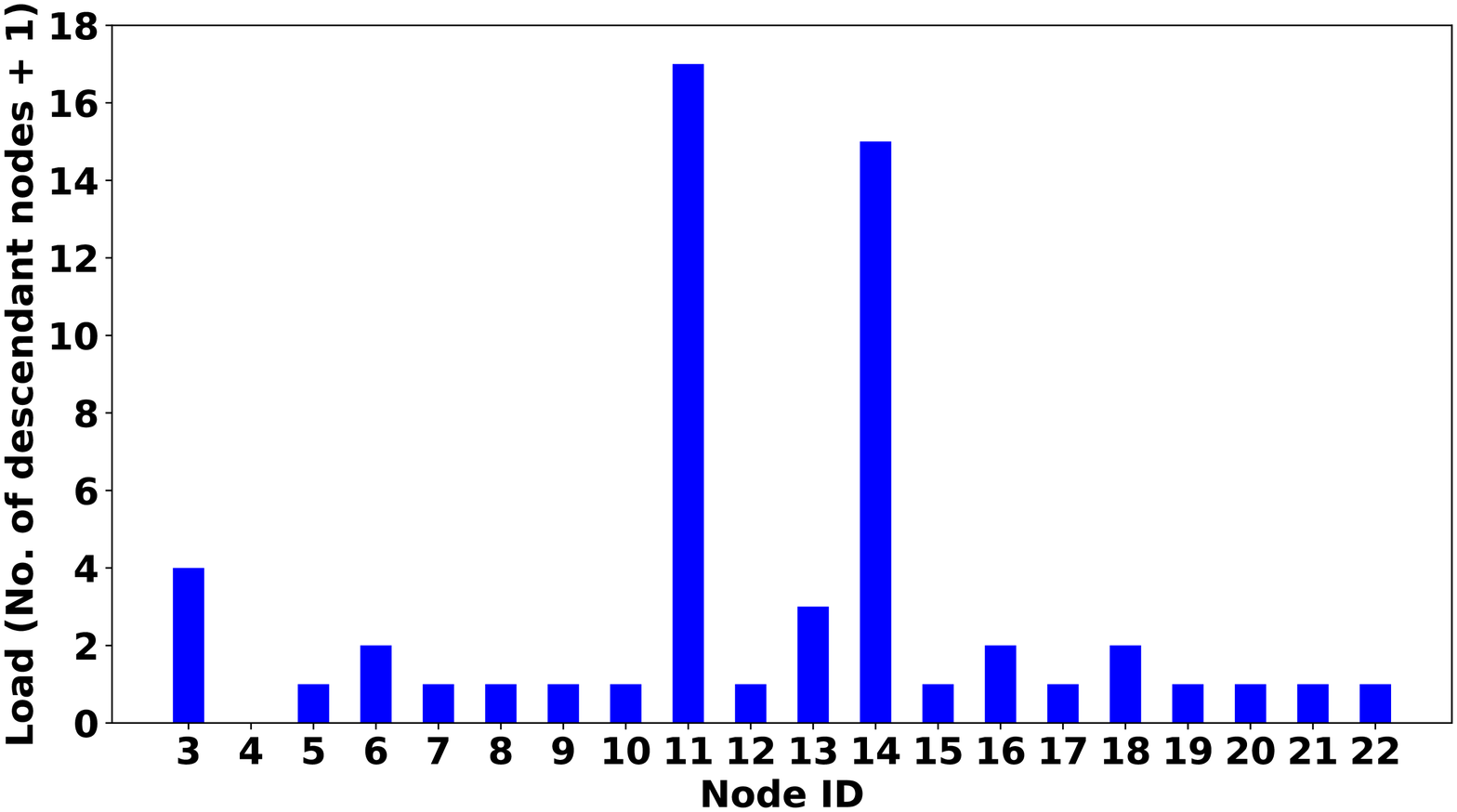}}
\hfil
\subfloat[RDCT - test case 3]{\includegraphics[scale=0.21]{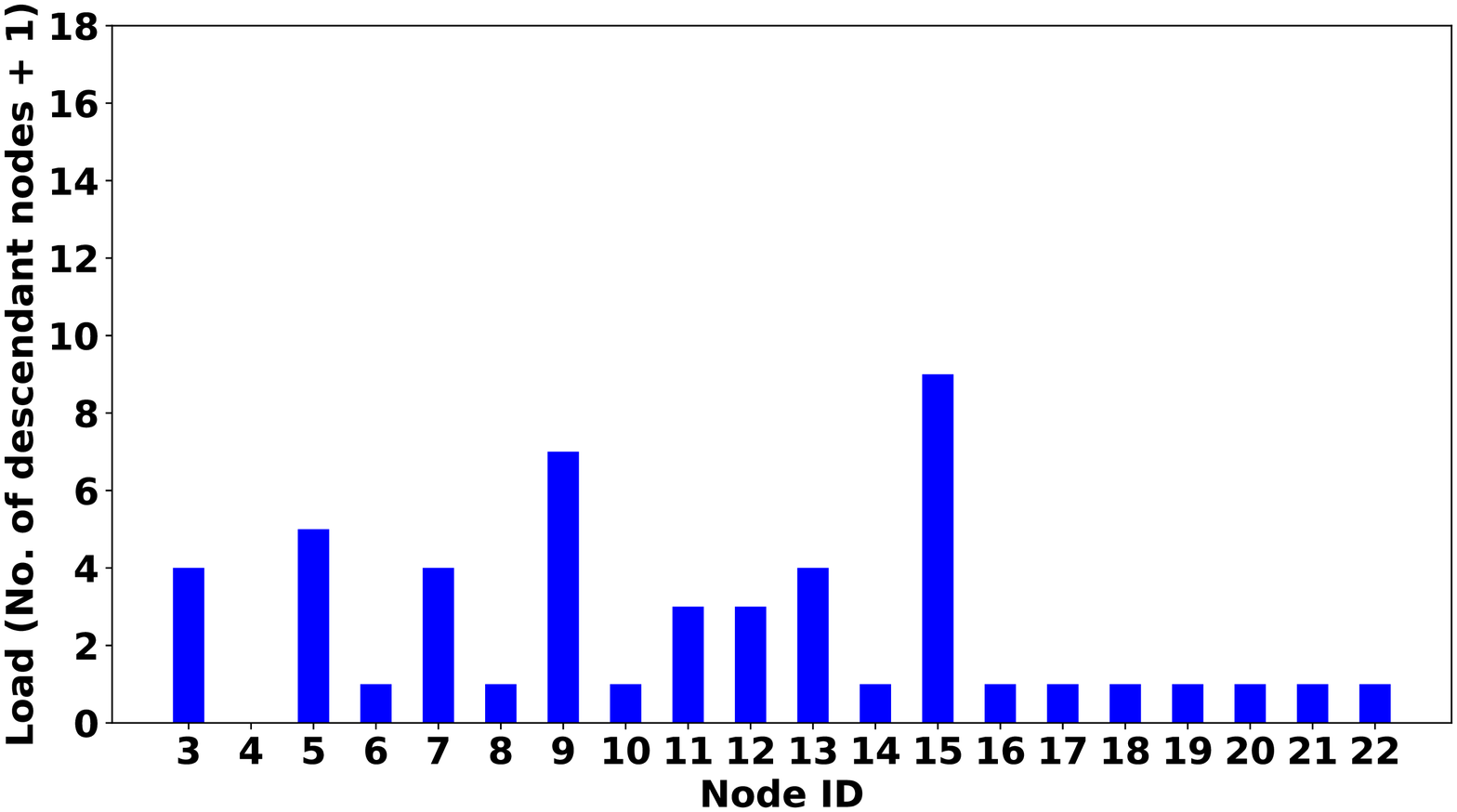}}
\caption{Load on each node during the first data collection round of each test case}
\label{load_0thDCT_eachtestcase}
\end{figure*}
 
 Fig.~\ref{load_0thDCT_eachtestcase} shows the load on each node in the first data collection tree for each test case. By load of a node in the tree, we mean the number of descendants the node has plus an additional one (representing its own generated data). Each node has a fully charged battery during this data collection tree formation stage and hence a high lifetime tree is expected to 
 have a highly balanced load. It is clear from Fig.~\ref{load_0thDCT_eachtestcase} 
 that \emph{BDCT has better load balancing when compared with SPT and RDCT}; RDCT performs the worst. This trend is consistently observed in all the test cases.

\begin{figure*}[h]
\centering
\subfloat[BDCT - $1^{st}$ data collection round]{\label{BDCT_round1}\includegraphics[scale=0.36]{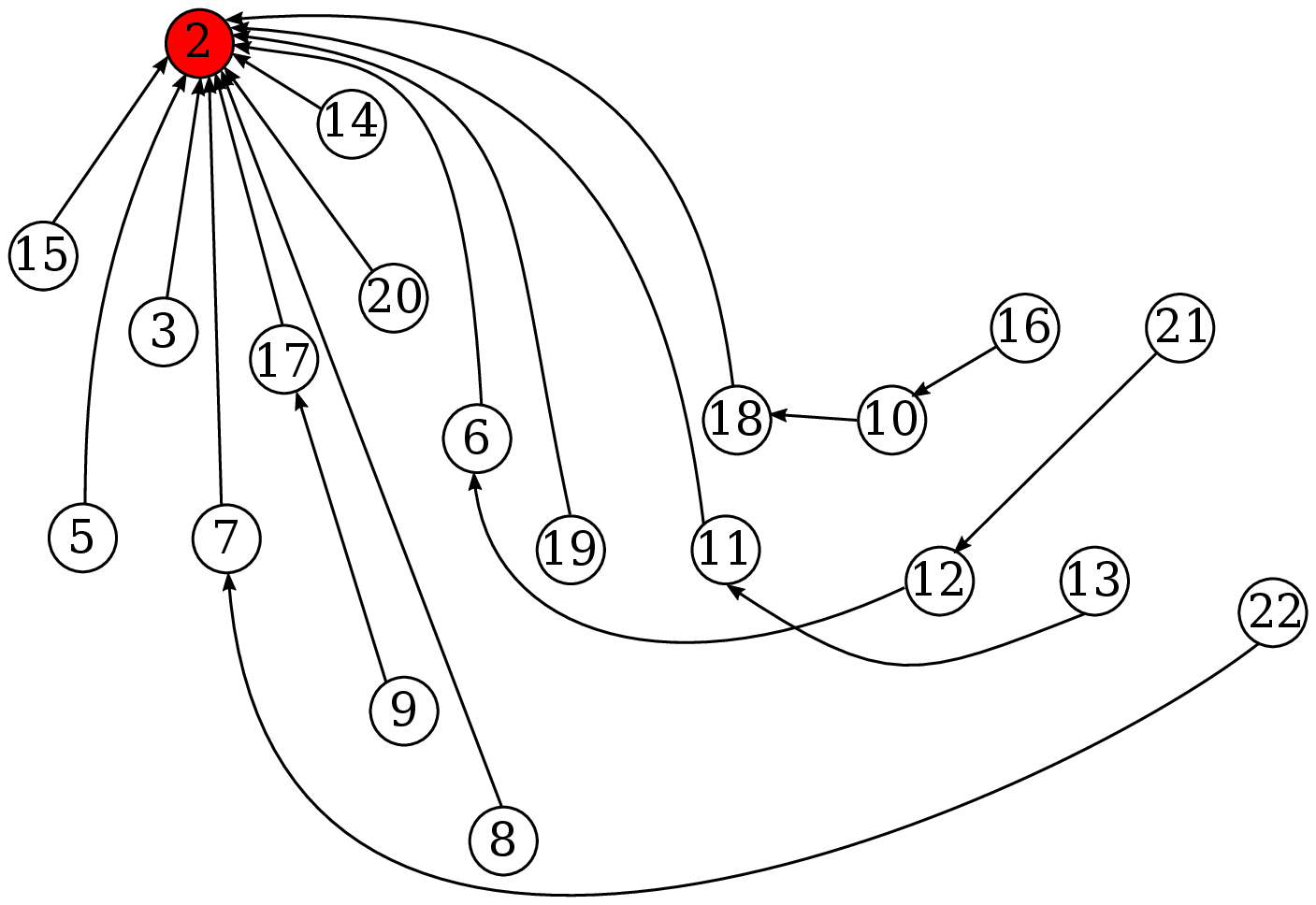}}
\hfil
\subfloat[BDCT - $4^{th}$ data collection round]{\label{BDCT_round4}\includegraphics[scale=0.36]{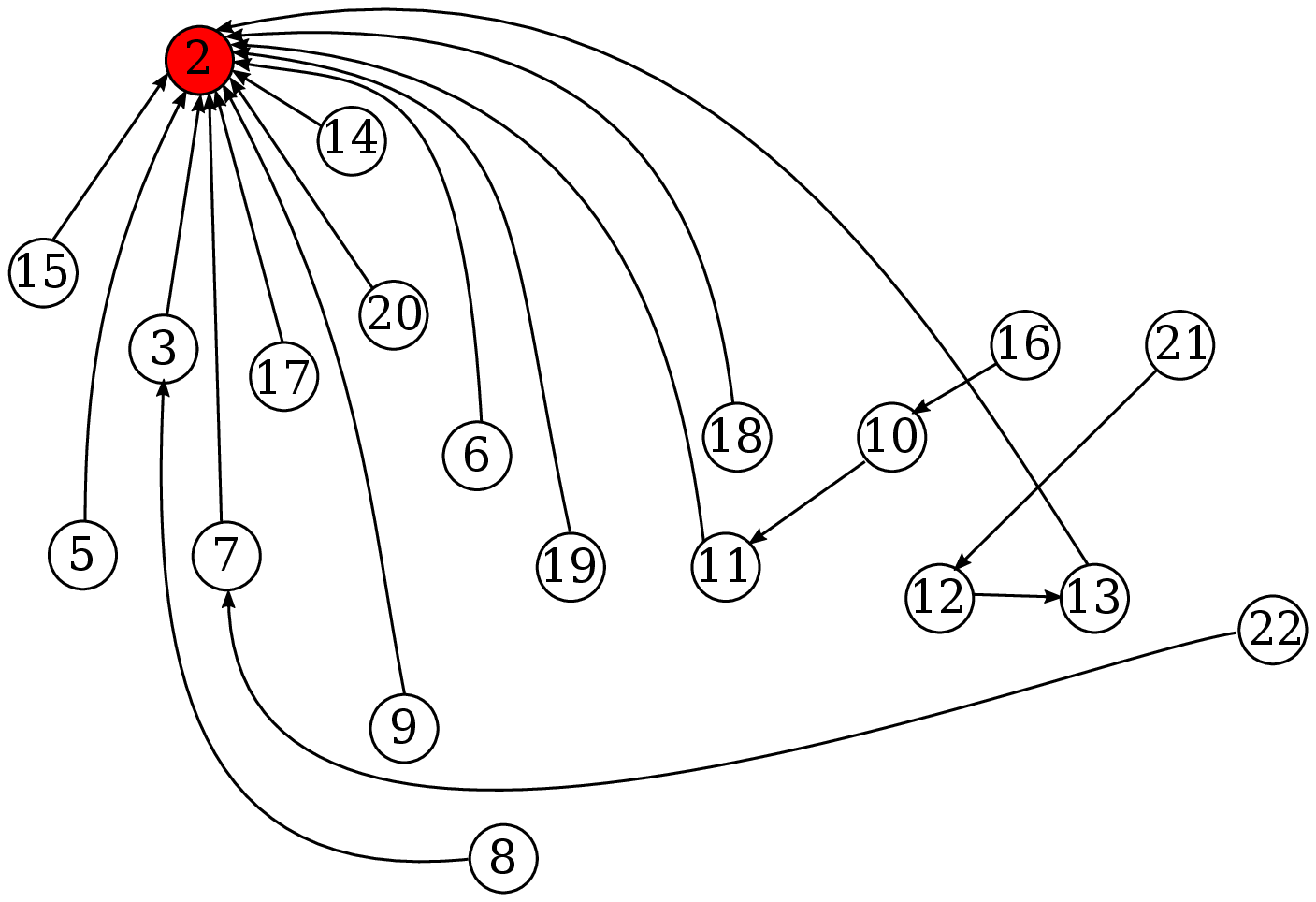}}
\hfil
\subfloat[BDCT - $7^{th}$ data collection round]{\label{BDCT_round7}\includegraphics[scale=0.36]{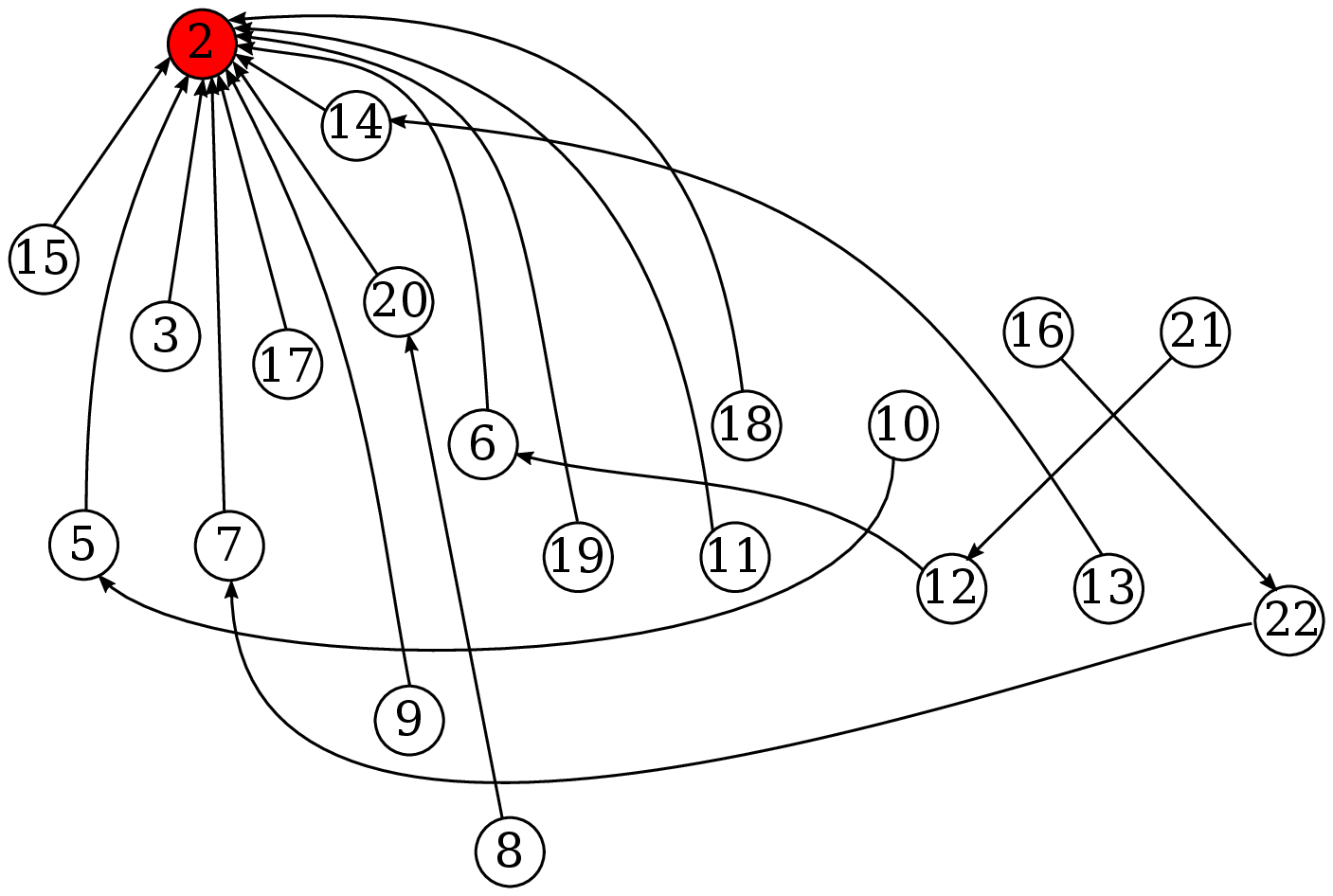}}
\hfil
\subfloat[SPT - $1^{st}$ data collection round]{\label{SPT_round1}\includegraphics[scale=0.36]{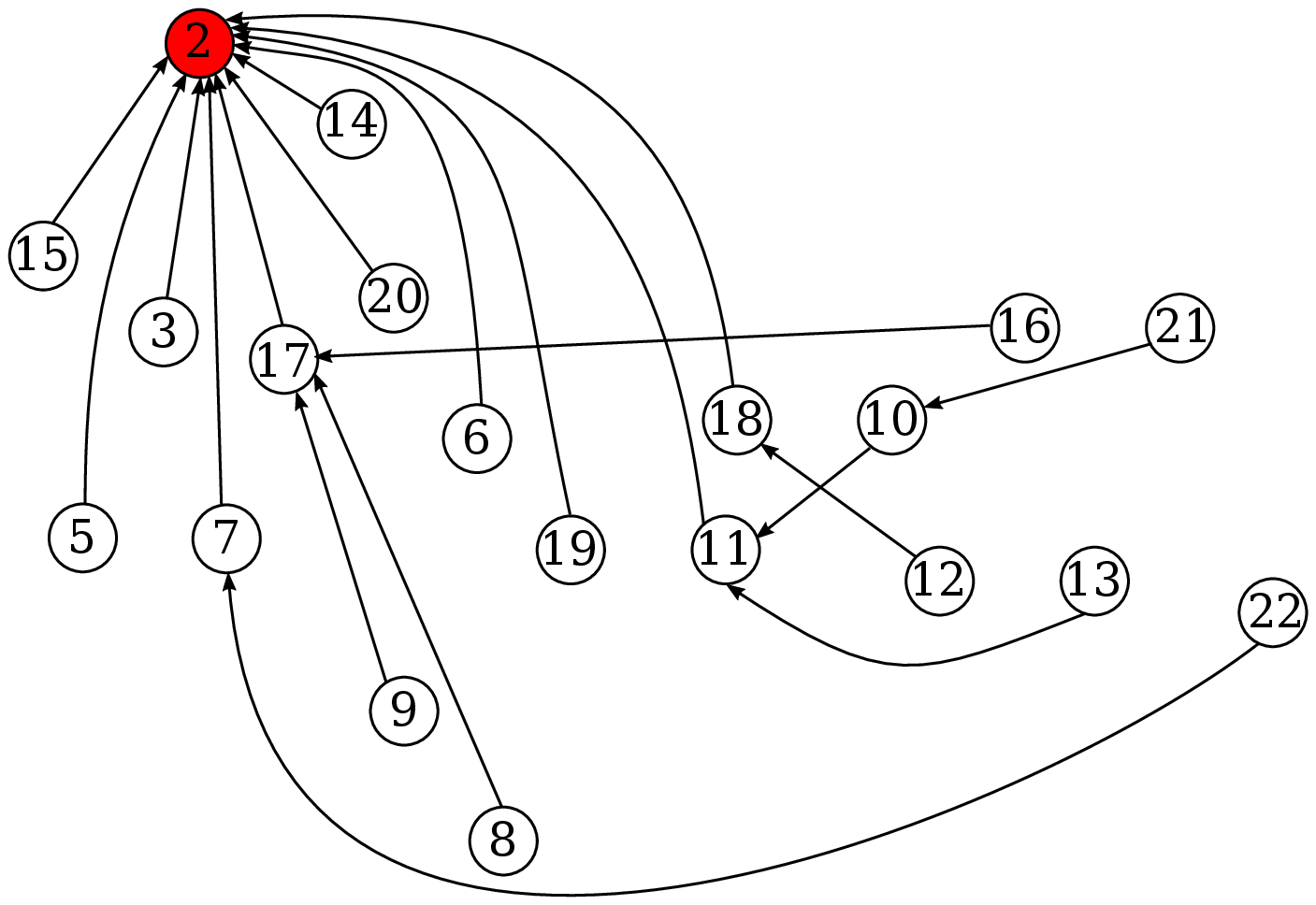}}
\hfil
\subfloat[SPT - $4^{th}$ data collection round]{\label{SPT_round4}\includegraphics[scale=0.36]{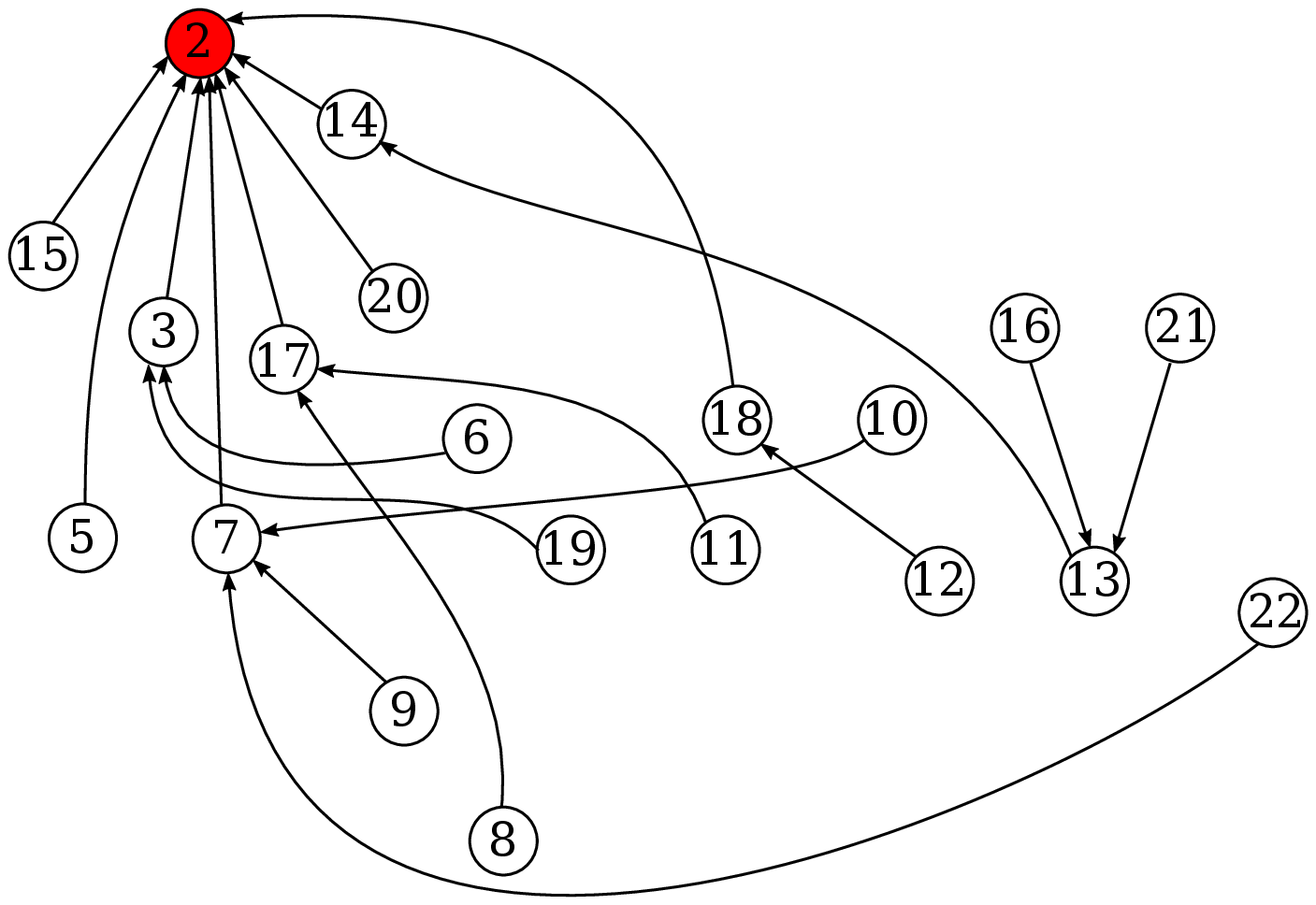}}
\hfil
\subfloat[SPT - $7^{th}$ data collection round]{\label{SPT_round7}\includegraphics[scale=0.36]{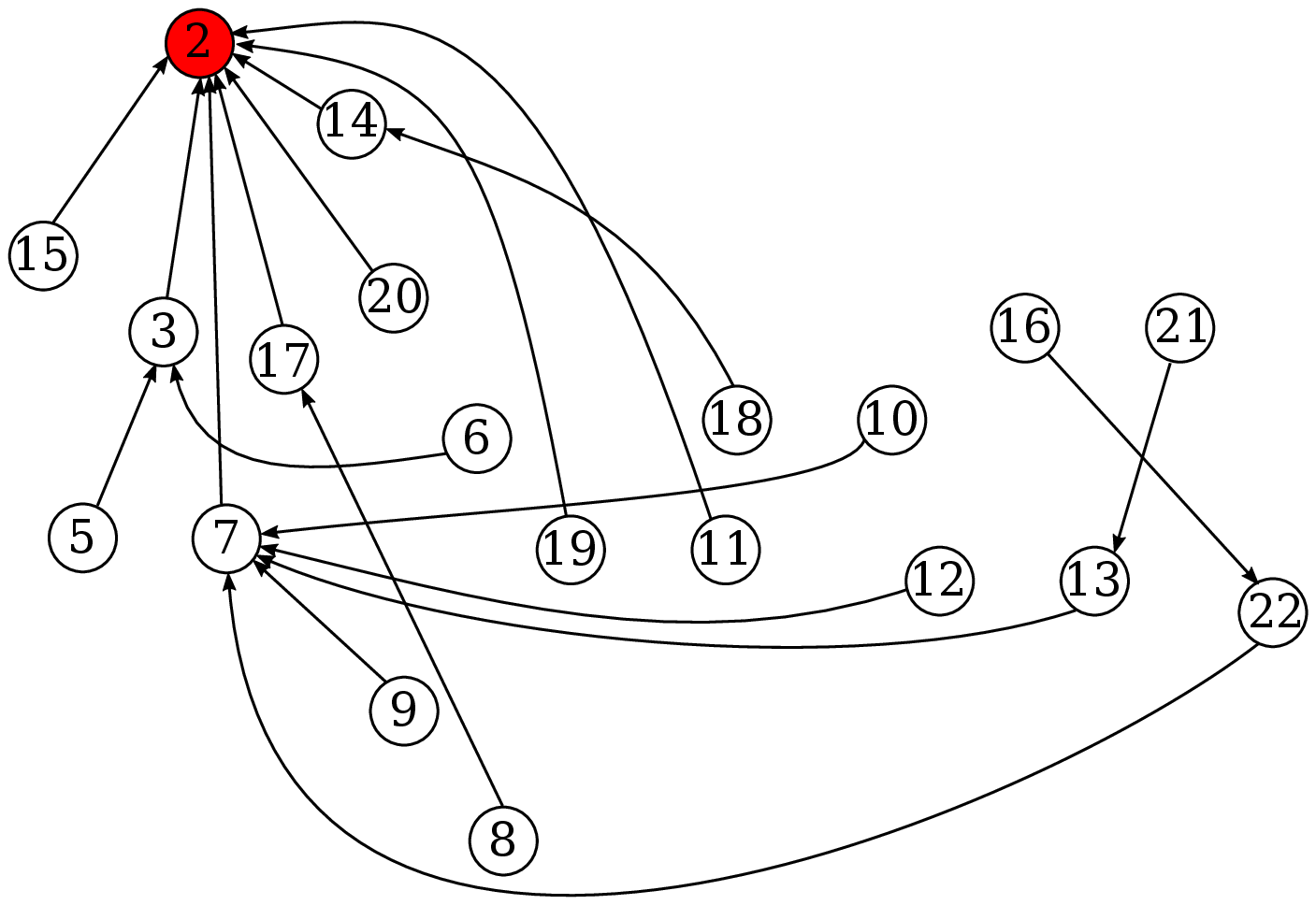}}
\hfil
\subfloat[RDCT - $1^{st}$ data collection round]{\label{RDCT_round1}\includegraphics[scale=0.36]{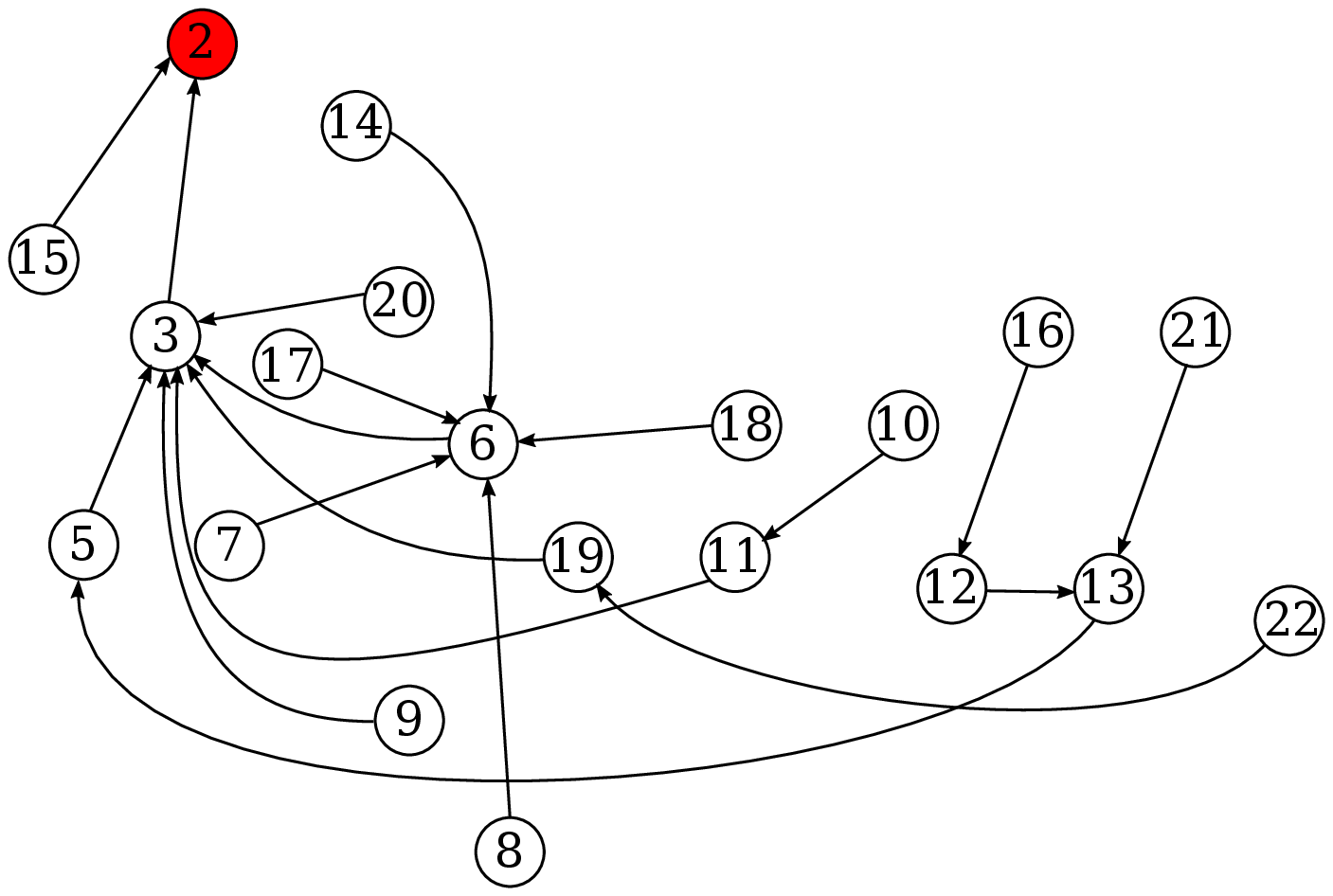}}
\hfil
\subfloat[RDCT - $4^{th}$ data collection round]{\label{RDCT_round4}\includegraphics[scale=0.36]{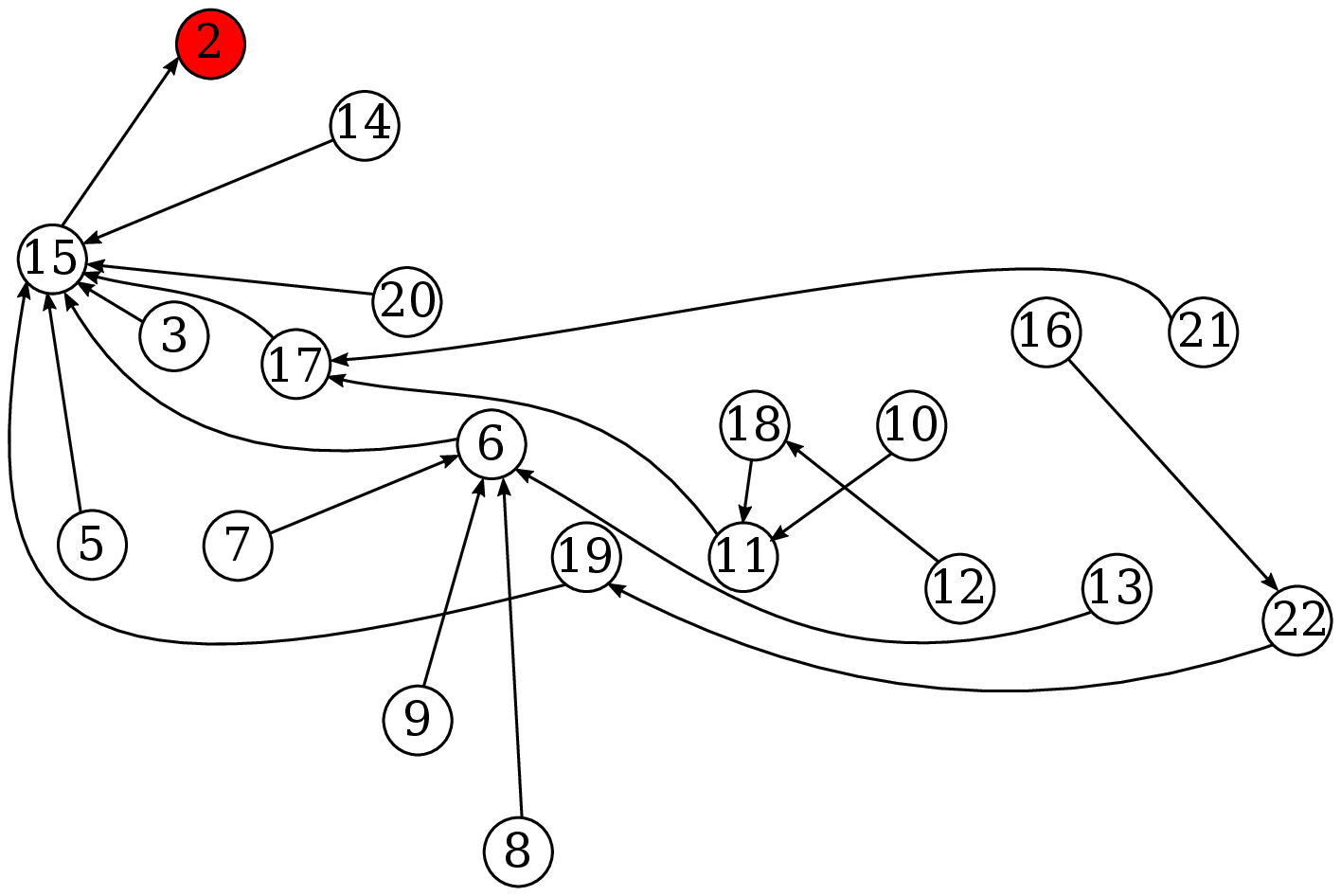}}
\hfil
\subfloat[RDCT - $7^{th}$ data collection round]{\label{RDCT_round7}\includegraphics[scale=0.36]{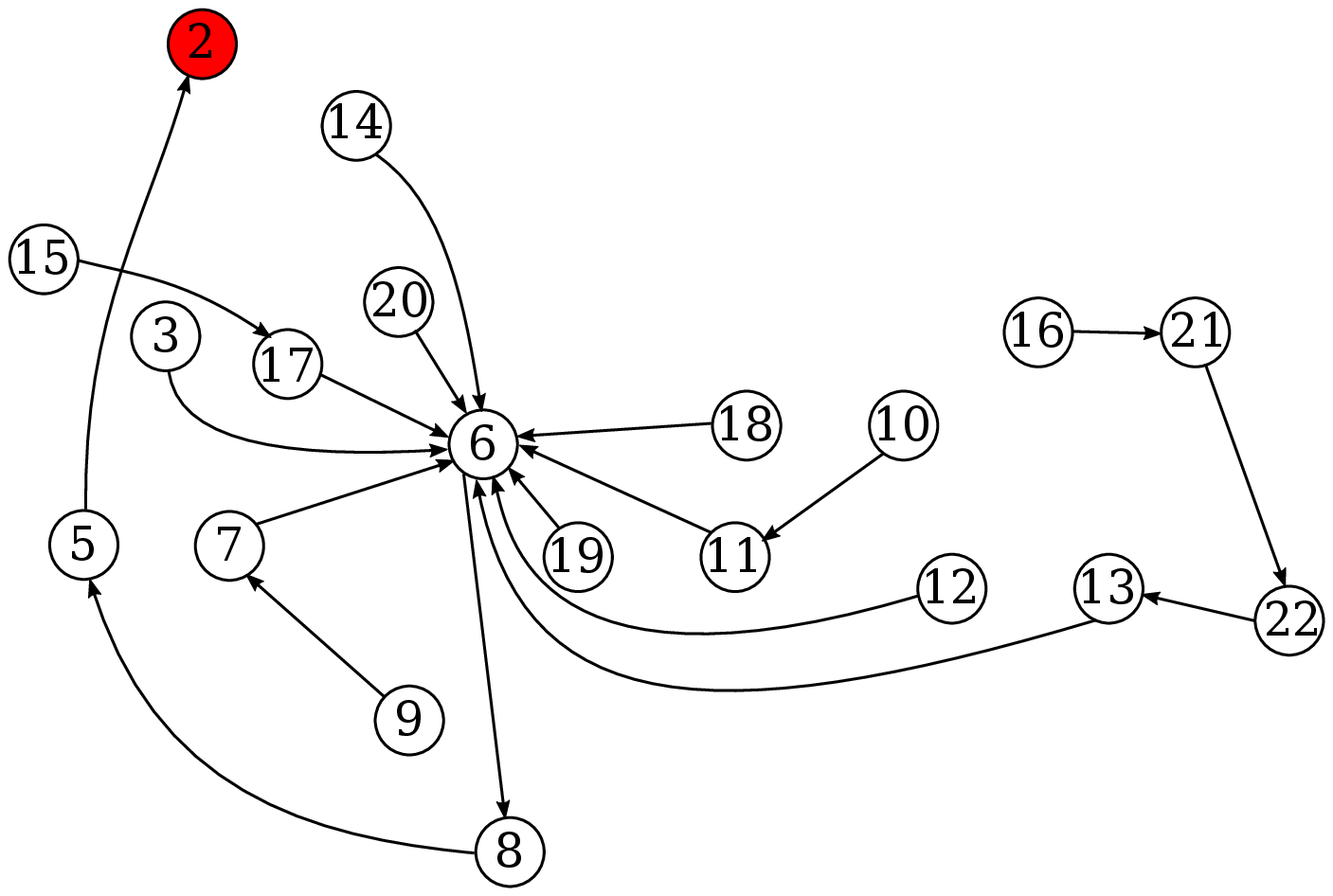}}

\caption{Data collection trees constructed during different data collection rounds in test case 1}
\label{single_test_case_different_dcts}
\end{figure*}

Fig.~\ref{single_test_case_different_dcts} shows some of the data collection trees constructed during different data collection rounds 
of test case 1 under each of the three algorithms. 
The major observations from Fig.~\ref{single_test_case_different_dcts} are as follows:

\begin{enumerate}

 \item {Many nodes directly get connected with the sink node under BDCT as well as under SPT. This is expected under BDCT since the root node (sink) is considered to have infinite energy. Also, it happens under SPT since the energy required for direct communication with the sink is less than that for multi-hop communication. This trend is not visible in case of RDCT. }

\item {The \emph{loads on different nodes are observed to be more balanced under BDCT than under SPT and RDCT}. For example, the maximum load of any node under 
BDCT is observed to be three (refer to node 18 in Fig.~\ref{BDCT_round1}, node 11 in Fig.~\ref{BDCT_round4} and node 6 in Fig.~\ref{BDCT_round7}).
Among the trees formed under SPT, node 7 in Fig.~\ref{SPT_round7} has a load of 8 and among the trees formed under RDCT,
node 15 in Fig.~\ref{RDCT_round4} and node  5 in Fig.~\ref{RDCT_round7} have the maximum load of 19 each.
}
\end{enumerate}

\subsection{Simulations based performance evaluation of the proposed algorithm in large networks}
\label{SSC:simulations}
In this section, we study the performance of the proposed algorithm in large networks through simulations using Python.
The lifetime of the network is considered as the time until the first node in the network fails
due to battery depletion.
The simulation studies are carried out to understand how the lifetime of the network varies in large networks 
when the data collection happens through different data collection trees. 
The network lifetime under the proposed algorithm (BDCT) is compared with those under the state-of-the-art Randomized
Switching for Maximizing Lifetime (RaSMaLai) algorithm \cite{imon2015energy}, the shortest path tree (SPT), minimum spanning tree (MST)~\cite{Dijkstra}, and random tree (RDCT) based data collection schemes.

The energy parameters used for the simulations are in correspondence with the actual energy consumption of various modules in the sensor nodes that are used in the experimental evaluations described in Section~\ref{SSC:testbed:based:evaluation}. For our simulations, we follow the same deployment strategies as used in the state-of-the-art work~\cite{imon2015energy}. In particular,  nodes
are randomly placed in an area of $100 m \times 100 m$, and the number of nodes ($N$) is varied from $50$ to $400$.
We have considered two test scenarios for the simulation studies. In scenario 1, the root node (sink) is placed at the center of the deployment area, whereas in scenario 2, the root node is placed at one corner of the deployment area. In both the test cases, at time $t=0$, the battery of each node is fully charged. Each node in the simulations supports different transmission power levels, and the maximum transmission range of a node is considered to be $25$ m.
``NetworkX''\cite{networkx} is a Python library that is widely used for studying graphs and networks.
To model the network, we have used the ``random-geometric-graph'' function \cite{networkx} from the ``NetworkX'' package for the graph generation.

Recall that the considered data collection approach (detailed in Section~\ref{Network_architecture}) involves two stages: formation of a data 
collection tree, followed by periodic data collection along the tree.
Once the data collection tree is constructed for a generated random connected graph $G$, the network enters into the periodic data collection phase. In this phase, every node updates its current remaining battery energy level in every data collection time slot by deducting the consumed energy from the remaining battery energy level at the end of the previous time slot. 
The energy expenditure of a sensor node in each data collection time slot includes the energy spent for sensor data generation,
data reception from its children and transmission to its parent node. The reconstruction of a data collection tree happens once every $k$ 
data collection time slots ($k = 10000$ in the simulations). The lifetime is defined to be the number of data collection time slots 
until one of the nodes in the network gets depleted of its battery energy.
The simulations ignore the energy expenditure of nodes during the 
sleep stage and the overhead which is required for the tree construction phase, since tree construction occurs very rarely (once every $10000$ data collection time slots).

\begin{figure}[]
\centering
  \subfloat[Test scenario~1: Root node at the centre]{\label{root_centre} \includegraphics[scale=0.2]{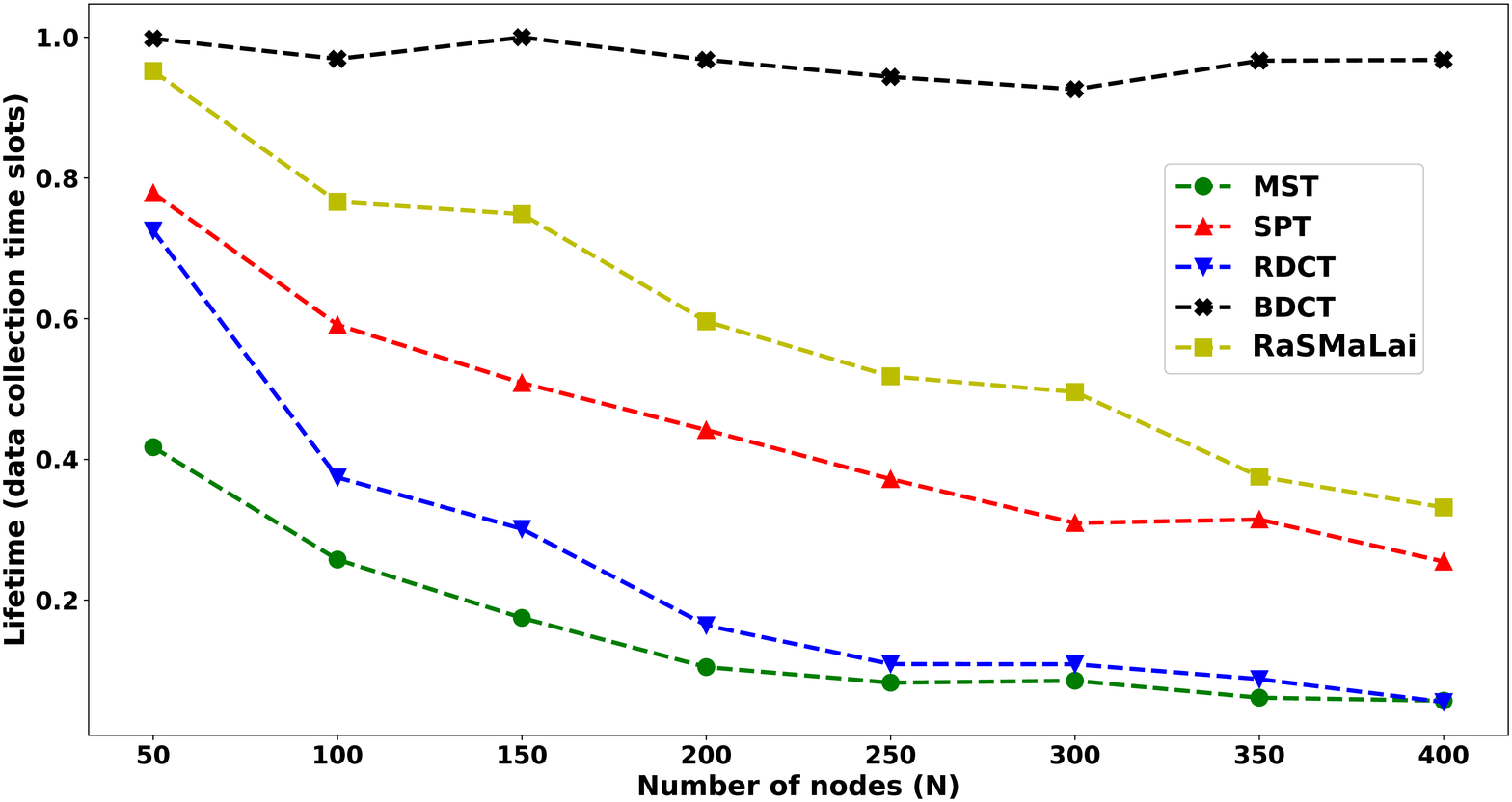}}
 \hfill
  \subfloat[Test scenario~2: Root node at a corner]{\label{root_corner} \includegraphics[scale=0.2]{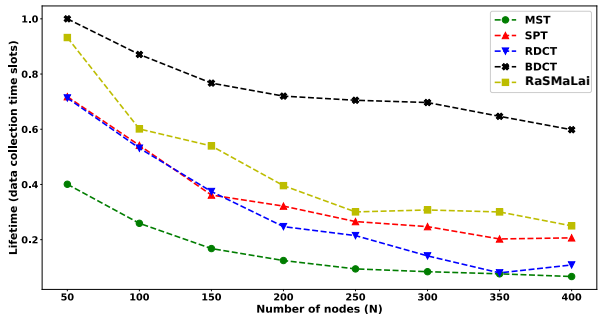}}
 \hfill
\caption{Lifetimes of the network when data collection occurs through trees constructed using BDCT, RaSMaLai, SPT, MST, and RDCT.}
\label{algo_simulations}
\end{figure} 

Fig. \ref{algo_simulations} shows the lifetimes of the network for both the test scenarios, for the cases when data collection happens through trees constructed using different algorithms, viz., BDCT, RaSMaLai~\cite{imon2015energy}, SPT, MST, and RDCT. Each lifetime evaluation was carried out for ten randomly generated graphs, and each point in Fig.~ \ref{algo_simulations} represents the average lifetime over the ten simulation trials. Also, the lifetime values are scaled to between $0$ and $1$ (normalized). 
Fig. \ref{algo_simulations} shows that \emph{the proposed data collection tree construction algorithm (BDCT) provides a significantly higher lifetime than all the other tree construction approaches in both the test scenarios}.  
Also, as the number of nodes ($N$) in the network increases, the network becomes denser, and the proposed algorithm (BDCT) outperforms the other data collection schemes (RaSMaLai, SPT, MST, and RDCT) by a larger margin. In particular, when $N = 400$, the network lifetime under the proposed BDCT algorithm is more than double that under the RaSMaLai, SPT, MST, and RDCT algorithms. 

\section{Conclusions}
\label{conclusion}
In this paper, we addressed the problem of building a maximum lifetime data collection tree for periodic data collection in sensor network applications.
We formulated the maximum lifetime data collection tree problem by considering the energy expenditure on a data \emph{packet} basis, in contrast to prior works, which consider it on a data \emph{unit} basis.
Variable transmission power levels of the radio and taking the sensor energy consumption into account are other factors that make our problem formulation different from those in prior work. We proved  NP-completeness of the formulated problem by  reducing the set cover problem to it and proposed a novel algorithm for finding a data collection tree with a high lifetime. The performance of the proposed algorithm was evaluated via its actual implementation on a WSN testbed consisting of 20 sensor nodes and 
compared with those of the SPT and RDCT algorithms. It was observed that the proposed algorithm discharges the nodes' battery voltages in a more balanced manner and thus provides a higher network lifetime than the SPT and RDCT algorithms. Also, we compared the performance of the BDCT algorithm in large networks with those of the RaSMaLai, MST, SPT and RDCT algorithms  through simulation studies. Our simulations show that the proposed BDCT algorithm provides a significantly higher lifetime than all the other tree construction approaches considered in both the test scenarios.

\section*{Acknowledgment}
Authors would like to acknowledge the financial support received for NNetRA (Nanoelectronics Network for Research and Application)
from Ministry of Electronics and Information Technology (MeitY), Department of Science \& Technology (DST) Government of India, through Centre of 
Excellence in Nanoelectronics, IIT-Bombay (Grant numbers: No. 5 (1)/2017-NANO and DST/NM/NNetRA/2018(G)-IIT-B). 
We also would like to thank Prof. Dinesh K. Sharma and Prof. Shabbir Merchant for the fruitful discussions and suggestions.

\section*{References}
\Urlmuskip=0mu plus 1mu
 \bibliographystyle{elsarticle-num} 
 \bibliography{Maximum_Lifetime}

\end{document}